\documentclass[journal]{IEEEtran}

\ifCLASSINFOpdf
\else
\fi

\usepackage{graphicx} 
\usepackage{epstopdf} 
\usepackage{mathptmx} 
\usepackage{times} 
\usepackage{amsmath} 
\usepackage{amsthm}
\usepackage{amssymb}  
\usepackage{txfonts}
\AtBeginDocument{\mathcode`v=\varv}
\usepackage{subfig}
\usepackage{cite}
\usepackage{color}
\usepackage{bm}
\usepackage{bbding,makecell}
\usepackage[font={small,sf}]{caption}

\DeclareMathOperator{\rank}{rank}
\newtheorem{theorem}{Theorem}
\newtheorem{remark}{Remark}

\newtheorem{definition}{Definition}
\newcommand{\Rmnum}[1]{\uppercase\expandafter{\romannumeral #1}}  

\begin{document}

\title{Correct Online Estimation of the Powertrain Time Constants in Adaptive Vehicular Platooning}

\author{Qiuhao Wen, Simone Baldi,~\IEEEmembership{Senior Member,~IEEE}, Jiwei Wang,\\  Wenwu Yu,~\IEEEmembership{Senior Member,~IEEE} and Di Liu,~\IEEEmembership{Member,~IEEE}
\vspace{-0.25cm}
\thanks{This work was partially supported by the National Natural Science Foundation of China under Grants  No. 62573115, 6251101378, 62233004, by Jiangsu Provincial Scientific Research Center of Applied Mathematics No. BK20233002, by UKRI Research and Innovation No. EP/Z002214/1, and by Horizon Europe Marie Sklodowska-Curie Action No. 101146446. (\textit{corresponding author: Di Liu})\newline
Q. Wen, S. Baldi, and W. Yu are with the Self-Organizing Mobility Control and Learning Lab, Southeast University, Nanjing, China. {\tt\small \{qiuhao.wen,s.baldi,wwyu\}@seu.edu.cn}\newline 
J. Wang is with Bernoulli Institute for Mathematics, Computer Science \& Artificial Intelligence, University of Groningen, The Netherlands. {\tt\small  jiwei.wang@rug.nl}\newline
D. Liu is with the Self-Organizing Mobility Control and Learning Lab, Southeast University, Nanjing, China, and also with the Department of Electrical and Electronic Engineering, Imperial College, London, United Kingdom.  {\tt\small di.liu@imperial.ac.uk}
}
}

%
%

\IEEEtitleabstractindextext{
\begin{abstract}
    In longitudinal platooning, some key sources of uncertainty are the powertrain time constants of the vehicles. 
    Because such time constants appear in the input matrix of the platooning dynamics, their correct estimation is either impractical with methods requiring persistence of excitation, or impossible with methods requiring the input matrix to be known. 
    This work proposes a novel adaptive longitudinal platooning method with correct estimation of the powertrain time constants. 
    To achieve correct estimation, the composite adaptive control framework and its stability analysis are suitably modified to handle the time constant uncertainty in the design of the adaptive law. 
    The result is a platooning protocol that guarantees convergence of the estimated time constants to their true values without the need for persistence of excitation: it is sufficient the derivative of the acceleration to be nonzero over a possibly short transient, an extremely relaxed excitation condition. 
    Comparisons with state-of-the-art platooning solutions reveal advantages such as no required measurements of acceleration derivative nor collection of past data. 
    The robustness and practicality of the proposed design is also verified with CarSim-based platooning experiments.
\vspace{-0.3cm}
\end{abstract}
\begin{IEEEkeywords}
Longitudinal platooning, connected vehicles, composite adaptive control, powertrain time constant, relaxed excitation condition. \vspace{-0.1cm}
\end{IEEEkeywords}
}

\maketitle

\thispagestyle{empty}
\pagestyle{empty}

\IEEEdisplaynontitleabstractindextext

%
\IEEEpeerreviewmaketitle

\section{Introduction}
Platoons of connected and automated vehicles have been judged as a promising solution to reduce traffic congestion and fuel consumption\cite{c18,c19,c34,c43,c46}.
Longitudinal platooning focuses on strings of vehicles coordinating in the traffic thanks to automated vehicle-following capabilities\cite{c20,c44}.
The objective is twofold: first, for each individual vehicle, guarantee stabilization of a desired inter-vehicle spacing policy\cite{c12,c35,c47}; second, for the whole platoon, guarantee rejection of disturbances that may propagate downstream along the vehicle string per effect of the vehicle-following interconnection.
The former is known as `individual stability', and the latter is known as `string stability'\cite{c22,c36}, for which different definitions have been proposed in the literature\cite{c21,c25,c37}. The vehicle dynamics in any platoon are heterogeneous and subject to uncertainty: 
platooning protocols based on robust control\cite{c67,c68,c69} may fail to attain individual and string stability over the whole uncertainty set. We focus on uncertainty scenarios requiring platooning protocols based on adaptive control, for which an overview of representative results is given hereafter.

\subsection{Related work on adaptive longitudinal platooning}
A key source of heterogeneity and uncertainty in longitudinal platooning is the powertrain time constant of each vehicle in the platoon. 
Different from systems modeled using first-, second-, or higher-order integrators\cite{c56,c57}, the powertrain time constant cannot be neglected in platooning systems, as it represents the response of the low-level control layer aiming to track a desired acceleration command\cite{c48}. 
Such time constant is also used to model the low-level control layer response in electric vehicles\cite{c59}.
As shown in~\cite{c61}, factors like vehicle type, mass, road topography, heavily affect the time constant, so that a platooning protocol must dynamically adapt to such time constant\cite{c47,c40,c45}.
While adaptive approaches based on neural networks\cite{c49,c70}, prescribed performance\cite{c51,c52} and adaptive consensus\cite{c53,c54} exist, a particularly convenient way to deal with heterogeneity and uncertainty of powertrain time constants is Model Reference Adaptive Control (MRAC)\cite{c7,c28}. 
The main idea of MRAC-based platooning is making the physical (actual) vehicle dynamics converge, through feedback control, to the dynamics of an ideal (virtual) vehicle which plays the role of a reference model chosen by the designer. 
Refer to Fig.~\ref{platoon} for a visual explanation.
Then, the adaptive platooning problem is split into two tasks: first, design a reference layer with desired stability properties (e.g., individual and string stability); second, by means of adaptive feedback\cite{c29}, let the physical layer track the reference layer so as to asymptotically attain the desired properties.
\begin{figure*}[thbp]
    \centering                        
    \includegraphics[width=1\linewidth]{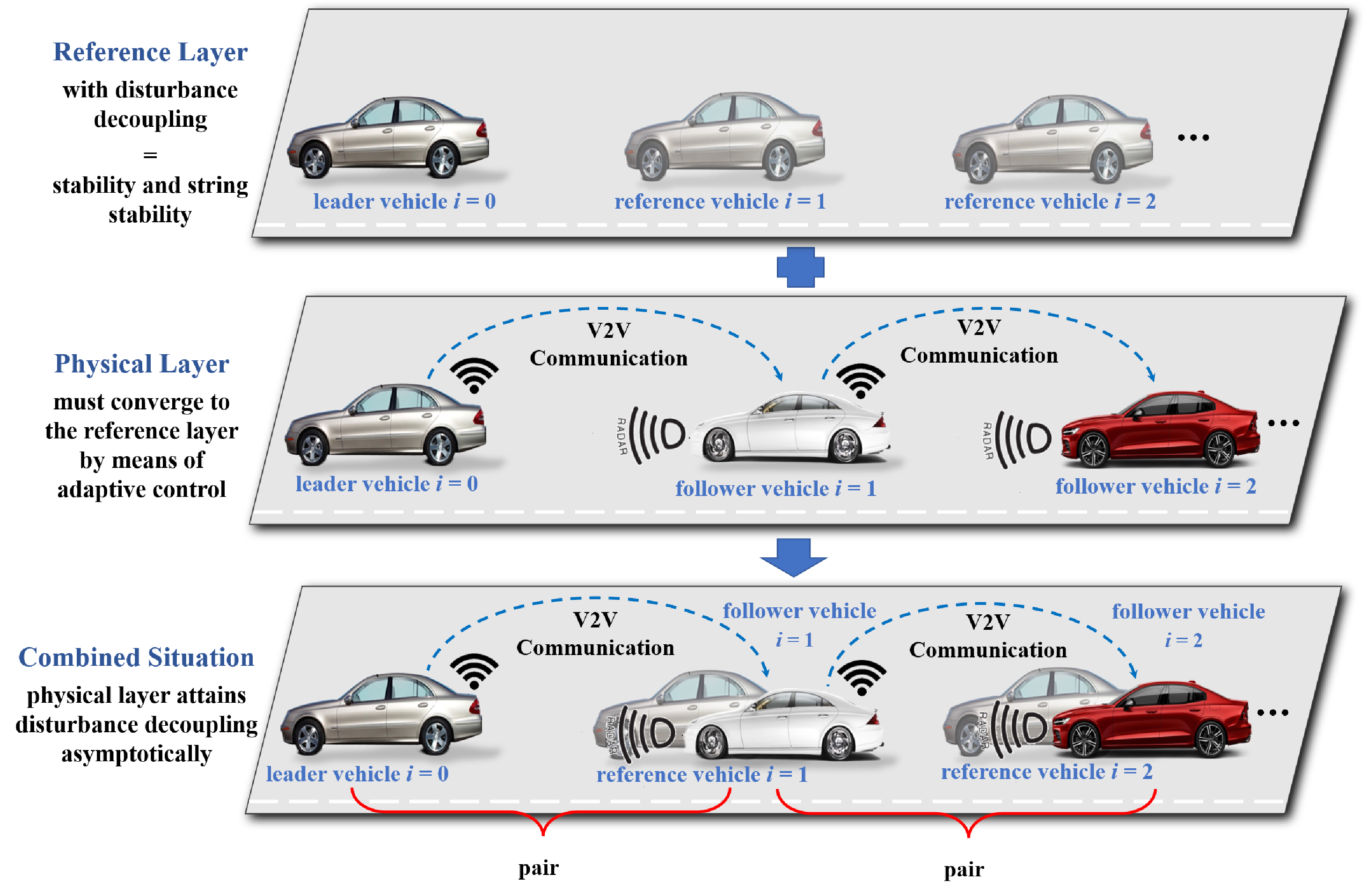}                          
    \caption{In cooperative platooning with the MRAC framework, there exist two layers: the reference layer and the physical layer.  
    We adopt a disturbance decoupling framework (cf. Remark 1) that allows to obtain both individual stability and string stability in a decoupled way, that is, by only considering predecessor-follower vehicle pairs.\label{platoon}}
\end{figure*}

Although adaptive platooning can guarantee asymptotic tracking of the reference layer, no guarantees can be given in general on the \emph{correct} estimation of the powertrain time constants for all vehicles, unless Persistence of Excitation (PE) conditions are satisfied\cite{c29}.
However, PE might be absent in many practical platooning scenarios, e.g., when the platoon proceeds at nearly-constant velocity (zero or almost zero acceleration): in these scenarios, asymptotic tracking of the reference layer is still achieved, but the estimated powertrain time constants may not converge to their \emph{correct} values. 
This has stimulated a search for adaptive control methods that could relax PE conditions.

\subsection{Related work on adaptive control with relaxed excitation}
Lack of convergence of the estimates to their true values gives rise to the well-known lack of robustness of adaptive control\cite{c29,c30}. 
The Concurrent Learning (CL) MRAC framework\cite{c31} is an approach introduced in the literature to guarantee correct estimation with relaxed PE conditions. 
In CL-MRAC, adaptation relies on both current and past data, and PE conditions are relaxed to rank conditions on such data.
However, CL-MRAC requires to measure extra state derivatives, namely, acceleration derivatives in the platooning case: it has been recognized that measuring such derivatives is unpractical\cite{c2}. 
Integral Concurrent Learning (ICL)\cite{c15} introduced a numerical integration to avoid such derivatives, with PE conditions being relaxed to integral rank conditions on data.
However, numerical integration can give rise to noise amplification. 
By replacing numerical integration with a suitably designed stable filter\cite{c32}, the Composite MRAC (C-MRAC) method removed the need for state derivatives and numerical integration, with PE conditions being relaxed to Finite Excitation (FE) conditions\cite{c17}.
However, as shown in\cite{c14}, the uncertainty in the powertrain time constants appears as an extra uncertainty term in the input matrix that cannot be handled by standard CL-MRAC, ICL-MRAC and C-MRAC designs.
This extra uncertainty term requires new tools both in the design of the adaptive law and in its stability analysis, leading to the main contributions of this study.

\subsection{Contribution of this study}
In this work we propose a novel platooning protocol in the framework of composite adaptive control. 
\begin{itemize}
    \item[1)] As compared to standard MRAC-based platooning designs\cite{c10,c33}, we relax PE conditions to FE conditions: to attain convergence of the estimated powertrain time constants to their true values, we only require 
		the derivative of the acceleration to be nonzero over a possibly short transient, a condition always met in practice, even in scenarios where the platoon converges to constant velocity;
    \item[2)] As compared to existing solutions based on concurrent learning\cite{c14}, the design we propose avoids the need for state derivatives (namely, acceleration derivatives) and avoids the need to collect past data;
    \item[3)] As compared to existing composite adaptive control methods\cite{c17}, we are able to include uncertainty in the input matrix, which allows to handle uncertainty in the powertrain time constants: we remark that existing methods are inapplicable to platooning because they require the input matrix to be known.
\end{itemize}
Extensive comparisons with the state-of-the-art designs clarify the advantages of the proposed protocol.
Its robustness and practicality 
are also verified with CarSim-based platooning experiments.

The rest of the paper is organized as follows: Section \Rmnum{2} gives the control problem statement.
After recalling state-of-the-art designs, the proposed composite adaptive platooning protocol is formulated in Section \Rmnum{3}.
Simulations and comparisons are in Section \Rmnum{4}, while Section \Rmnum{5} concludes the work.

{\it Notations}: 
$A^{\top}$ denotes the transpose of matrix $A$, $\bf I$ and $\bf 0$ denote the identity matrix and zero matrix of appropriate dimensions, ${\rm tr}(A)$ denotes the trace of matrix $A$. The 2-norm of vector $x$ is denoted by $\Vert x \Vert$,
$\lambda_{\min}(A)$ and $\lambda_{\max}(A)$ denote the smallest and the largest eigenvalue of matrix $A$, respectively.

\section{Control Problem Statement}
In line with the literature~\cite{c1,c2,c3,c4,c5}, the longitudinal platooning dynamics for a predecessor-follower vehicle system (the vehicles are indexed as $p$ and $f$, respectively) are expressed as:
\begin{equation}\label{standard dyanamics}
    \begin{aligned}
        \dot{s}_p(t)&=v_p(t),\\
        \dot{v}_p(t)&=a_p(t),\\
        \tau_p\dot{a}_p(t)&=-a_p(t)+u_p(t),
    \end{aligned}\quad
    \begin{aligned}
        \dot{s}_f(t)&=v_f(t),\\
        \dot{v}_f(t)&=a_f(t),\\
        \tau_f\dot{a}_f(t)&=-a_f(t)+u_f(t),
    \end{aligned}
\end{equation}
where $s_f(t),\ v_f(t),\ a_f(t)$ represent the longitudinal position, velocity, acceleration of the follower vehicle (and similar for the predecessor vehicle $p$), $u_f(t)$ is the desired acceleration representing the control input and $\tau_f>0$ is the powertrain time constant. 
Refer to~\cite{c61} for explanations on how road slope, rolling resistance and aerodynamic drag enter in (\ref{standard dyanamics}) from the low-level control layer and end up affecting the powertrain time constant in an uncertain way.
As the exact knowledge of $\tau_f$ is hardly available~\cite{c7,c8,c9},  all powertrain time constants in the platoon will be taken unknown in this work.

To establish the longitudinal platooning task, the spacing error is calculated based on the actual inter-vehicle distance $d_f(t)$ and the desired inter-vehicle distance $d_f^*(t)$, defined as
\begin{equation}
    e_f(t)=d_f(t)-d_f^*(t)=s_p(t)-s_f(t)-hv_f(t),\label{spacing error}
\end{equation}
where $h>0$ is the time headway. 
To represent the dynamics of the predecessor-follower system, let us define the state $x(t)={[e_f(t)\ \nu_f(t)\ a_f(t)]}^\top\in\mathbb{R}^3$, where $\nu_f(t)=v_p(t)-v_f(t)$ is the relative velocity.
Based on (\ref{standard dyanamics}), the state dynamics can be written as
\begin{equation}
    \dot{x}(t)=A x(t) + B u_f(t) + G a_p(t),
    \label{closed loop dynamics}
\end{equation}
with 
\begin{equation*}
    A=\begin{bmatrix}
        0 &1 &-h\\
        0 &0 &-1\\
        0 &0 &-\tau_f^{-1}
    \end{bmatrix},\ \
    B=\begin{bmatrix}
        0\\
        0\\
        \tau_f^{-1}
    \end{bmatrix},\ \
    G=\begin{bmatrix}
        0\\
        1\\
        0
    \end{bmatrix}.
    \label{parameters of the closed loop dynamics}
\end{equation*}
Note that the predecessor acceleration $a_p(t)$ (or, indirectly, the predecessor input $u_p(t)$) enters the dynamics (\ref{closed loop dynamics}) as an exogenous term, and thus plays a similar role as a disturbance.
Next, the adaptive platooning problem will be formulated in the framework of MRAC. 

\subsection{Adaptive Platooning as MRAC}
The basic idea behind MRAC is to make the state of the system (in our case, the predecessor-follower system (\ref{closed loop dynamics})) follow the one of a reference model.

To design a reference model with desired properties, a possible way is to define a `virtual' following vehicle, indexed as $\bar{f}$,
\begin{equation}\label{reference model ogn}
    \begin{aligned}
        \dot{s}_{\bar{f}}(t)&=v_{\bar{f}}(t),\\
        \dot{v}_{\bar{f}}(t)&=a_{\bar{f}}(t),\\
        \tau_{\bar{f}} \dot{a}_{\bar{f}}(t)&=-a_{\bar{f}}(t)+u_{\bar{f}}(t),
    \end{aligned}
\end{equation}
with desired $\tau_{\bar{f}}>0$.

Similar to (\ref{spacing error}), a reference spacing error and a reference relative velocity can be defined as $e_{\bar{f}}(t)=s_p(t)-s_{\bar{f}}(t)-hv_{\bar{f}}(t)$ and $\nu_{\bar{f}}(t)=v_p(t)-v_{\bar{f}}(t)$, leading to the reference state $\bar{x}(t)={[e_{\bar{f}}(t)\ \nu_{\bar{f}}(t)\ a_{\bar{f}}(t)]}^\top$.
The literature has shown that a suitable $u_{\bar{f}}(t)$ for the virtual vehicle can be designed in the disturbance decoupling framework~\cite{c10}, taking the form 
\begin{equation}
    u_{\bar{f}}(t)=[\theta_1\ \ \theta_2\ \ 1\hspace{-0.075cm}-\hspace{-0.075cm}\tau_{\bar{f}} h^{-1}\hspace{-0.1cm}-\hspace{-0.075cm}h\theta_2]\bar{x}(t)+\tau_{\bar{f}} h^{-1}a_p(t),\label{reference model input}
\end{equation}
where $\theta_1,\theta_2>0$.
The reference model is obtained from the closed loop of the virtual vehicle $\bar{f}$ and its controller (\ref{reference model input}), leading to
\begin{equation}
    \dot{\bar{x}}(t) =\bar{A}\bar{x}(t) + \bar{G}a_p(t),
    \label{reference model}
\end{equation}
where
\begin{equation}
    \bar{A}=\begin{bmatrix}
        0 &1 &-h\\
        0 &0 &-1\\
        \theta_1\tau_{\bar{f}}^{-1} &\theta_2\tau_{\bar{f}}^{-1} &-h^{-1}\hspace{-0.075cm}-\hspace{-0.075cm}h\theta_2\tau_{\bar{f}}^{-1}
    \end{bmatrix},\ \
    \bar{G}=\begin{bmatrix}
        0\\
        1\\
        h^{-1}
    \end{bmatrix}.
    \label{reference model matrices}
\end{equation}
Note that $\bar{A}$ is Hurwitz for any $\tau_{\bar{f}}, \theta_1, \theta_2 > 0$. 
Similar to (\ref{closed loop dynamics}), $a_p(t)$ enters in (\ref{reference model}) as an exogenous disturbance term.
The controller (\ref{reference model input}) takes the name of disturbance decoupling controller as it lets the reference model (\ref{reference model}) enjoy system properties known in the literature~\cite{c12,c13} and remarked hereafter.
\begin{remark}[Properties of the reference model]\label{Remark of properties of the reference model}
    The reference dynamics (\ref{reference model})-(\ref{reference model matrices}) are such that:
    \begin{enumerate}
        \item For any bounded $u_p(\cdot)$, the following \emph{disturbance decoupling output stability} property holds
        \begin{equation*}
            \lim_{t\to\infty}e_{\bar{f}}(t)=0;
        \end{equation*} 
        \item For any bounded $u_p(\cdot)$, the following \emph{string stability} property holds
        \begin{equation}\label{L2 string stability}
           \lim_{t\to\infty}\int_t^{t+T}(\lvert \nu_{\bar{f}}(\ell) \rvert ^2 -\lvert \nu_p(\ell) \rvert ^2) d \ell \le 0,
        \end{equation}
        for all $T>0$.
    \end{enumerate}
\end{remark}

The first item states that the spacing error converges to zero independently of the predecessor driving strategy: in other words, the spacing error will be decoupled from the effect of the disturbance.
The second item defines string stability in the sense of $\mathcal{L}_2$ norm on a certain interval. 
By interpreting the $\mathcal{L}_2$ norm as the energy of a signal, (\ref{L2 string stability}) states that the energy of $\nu_{\bar{f}}$ is not larger than the energy of $\nu_p$: as a result, the effect of a disturbance will not be amplified from the predecessor to the (virtual) follower.

Next, we aim to design an appropriate controller $u_f(t)$ so that the state $x(t)$ can track the reference state $\bar{x}(t)$ and converge to similar properties as in Remark~\ref{Remark of properties of the reference model}. 
To elaborate on the design of such controller, $\tau_f$ is assumed known first.
If $\tau_f$ is known, consider the ideal controller
\begin{equation}\label{ideal controller}
    \begin{aligned}
        u_f^*(t) =& a_f(t)+\tau_f h^{-1} a_p(t)\\
        &+\tau_f[\theta_1\tau_{\bar{f}}^{-1}\ \theta_2\tau_{\bar{f}}^{-1}\ -\hspace{-0.075cm}(h^{-1}\hspace{-0.075cm}+\hspace{-0.075cm}h\theta_2\tau_{\bar{f}}^{-1})] x(t),        
    \end{aligned}
\end{equation}
where $u_f^*(t)$ stands for an ideal version of $u_f(t)$.
In the following, for brevity, let us denote 
\begin{equation}\label{regressor}
    \phi(t)=[\theta_1\tau_{\bar{f}}^{-1}\ \theta_2\tau_{\bar{f}}^{-1}\ -\hspace{-0.075cm}(h^{-1}\hspace{-0.075cm}+\hspace{-0.075cm}h\theta_2\tau_{\bar{f}}^{-1})] x(t) + h^{-1} a_p(t).    
\end{equation}
The controller (\ref{ideal controller}) is ideal because by substituting it into (\ref{closed loop dynamics}), together with (\ref{reference model}) and (\ref{reference model matrices}), we have
\begin{equation}
    \dot{\tilde{x}}(t)=\bar{A}\tilde{x}(t),
    \label{ideal tracking error}
\end{equation}
where $\tilde{x}(t)=x(t)-\bar{x}(t)$ is the tracking error between the closed-loop system and the reference model. 
In (\ref{ideal tracking error}), $\tilde{x}(t)$ converges to zero as $\bar{A}$ is Hurwitz,
which means that the closed-loop system converges to the same behavior of the reference model.

However, as $\tau_f$ in (\ref{ideal controller}) is unknown, let us propose, instead of (\ref{ideal controller}), an adaptive controller given by
\begin{equation}
    u_f(t)=a_f(t)+\hat{\tau}_f(t)\phi(t),
    \label{actual controller}
\end{equation}
with $\hat{\tau}_f(t)$ being an estimate of $\tau_f$.
The dynamics of $\tilde{x}(t)$, resulting from (\ref{closed loop dynamics}) and (\ref{reference model}) with controller (\ref{actual controller}), are
\begin{equation}
    \dot{\tilde{x}}(t)=\bar{A}\tilde{x}(t)+\tilde{B}\frac{\tilde{\tau}_f(t)}{\tau_f}\phi(t),
    \label{actual tracking error}
\end{equation}
where $\tilde{\tau}_f(t)=\hat{\tau}_f(t)-\tau_f$ is the estimation error and $\tilde{B}^\top=[ 0\ 0 \ 1]$.

Note that (\ref{actual tracking error}) contains the same ideal dynamics of (\ref{ideal tracking error}), plus a disturbance term given by the estimation error. 
The problem is to design an adaptive law for $\hat{\tau}_f(t)$ to eliminate the disturbance in (\ref{actual tracking error}).

The standard MRAC\cite{c10} adjusts the estimate $\hat{\tau}_f(t)$ via the following adaptive law:
\begin{equation}\label{standard MRAC adaptive law}
    \dot{\hat{\tau}}_f(t)=-\gamma \tilde{B}^\top P\tilde{x}(t)\phi(t),
\end{equation}
where $\gamma>0$ is a learning gain and $P=P^\top$ is the positive definite solution of the Lyapunov equation $\bar{A}^\top P+P\bar{A}+Q=0$ with $Q>0$ being an arbitrary positive definite matrix.
Our goal is to propose a different solution than (\ref{standard MRAC adaptive law}), based on the problem formulation introduced hereafter.

\subsection{Control Problem}
Before formulating the control problem, the following different definitions about exciting signals are introduced.
\begin{definition}[Persistent Excitation (PE)\cite{c29}]\label{PE definition}
    Given a signal $\Phi(\cdot) \in\mathbb{R}^m$, if there exist $T > 0$ and $\alpha > 0$ such that
    \begin{equation*}
        \int_{t}^{t+T} \Phi(\ell)\Phi^\top(\ell) d\ell\ge\alpha {\bf I}>{\bf 0},\quad \forall t\ge 0,      
    \end{equation*}
    then the signal $\Phi(\cdot)$ satisfies persistent excitation.
\end{definition}
\begin{definition}[Rank Excitation (RE)\cite{c31}]\label{RE definition}
Given a set of recorded data
 $$Z = [\Phi(t_1),\dots,\Phi(t_k)] \in \mathbb{R}^{m \times k},$$
 if $\rank(Z) = m$ (equivalently, $Z Z^\top \geq \alpha {\bf I} > 0$ for some $\alpha>0$), then the recorded data satisfies rank excitation.
\end{definition}
\begin{definition}[Integral Rank Excitation (IRE)\cite{c15}]
Given a set of recorded data $Z = [\Phi(t_1),\dots,\Phi(t_k)] \in \mathbb{R}^{m \times k}$, if there exist $\Delta t>0$ and $\alpha>0$, such that $\bar{Z}\bar{Z}^\top\ge\alpha {\bf I}>0$, where 
$$\bar{Z}=[{\int_{\max\{t_1-\Delta t, 0\}}^{t_1}\Phi(\ell)}d\ell, \dots, {\int_{\max\{t_k-\Delta t, 0\}}^{t_k}\Phi(\ell)}d\ell],$$
then the recorded data satisfies integral rank excitation.
\end{definition}    
\begin{definition}[Finite Excitation (FE)\cite{c17}]\label{FE definition}
    Given a signal $\Phi(\cdot) \in\mathbb{R}^m$, if there exist time instants $t_e>t_s>0$, and $\alpha > 0$ such that
    \begin{equation*}
        \int_{t_s}^{t_e} \Phi(\ell)\Phi^\top(\ell) d\ell\ge\alpha {\bf I}>{\bf 0},     
    \end{equation*}
    then the signal $\Phi(\cdot)$ satisfies finite excitation over $[t_s, t_e]$.
\end{definition}

It is straightforward that Persistent Excitation implies Finite Excitation. 
Meanwhile, it is not difficult to show that Rank Excitation implies Integral Rank Excitation: this holds because $\rank(Z)=m$ implies $Z Z^\top \ge \alpha {\bf I}$ for some $\alpha >0$, a special case of Integral Rank Excitation for $\Delta t \rightarrow  0$.
Let us also mention that Finite Excitation is closely related to Interval Excitation (IE)\cite{c65,c66,c71,c72}: we do not cover such approaches\cite{c72,c71} because their problem formulation is based on obtaining scalar regressor models, but $\phi$ is scalar by design here.
The following control problem is considered in this work.

\textbf{Control Problem:}
Consider the predecessor-follower model (\ref{standard dyanamics}) with spacing error (\ref{spacing error}).
The control problem is to design an adaptive law for $\hat{\tau}_f(t)$ in the controller $u_f(t)$ (\ref{actual controller}) such that, for any unknown $\tau_f,\ \tau_p$ and any bounded $u_p(\cdot)$
\begin{equation}\label{x convergence}
\lim_{t \rightarrow \infty} x(t)-\bar{x}(t) = 0,
\end{equation}
which implies, from Remark 1, tracking a stable and string stable behavior.
In addition, the adaptive controller should guarantee convergence of the estimates, i.e.,
\begin{equation}\label{tau convergence}
\lim_{t \rightarrow \infty} \hat{\tau}_f (t)-\tau_f = 0,
\end{equation}
without PE and without extra state measurements as compared to those needed in the control law. $\hfill \diamond$

None of the existing adaptation solutions solves the control problem: namely, the standard MRAC law (\ref{standard MRAC adaptive law}) requires PE to guarantee that the estimate converges to the actual $\tau_f$\cite{c29} and thus cannot solve the control problem.
Meanwhile, CL-MRAC can relax PE to RE at the price of requiring extra state measurements\cite{c14}, thus also not solving the control problem. 

\section{Composite Adaptive Platooning}
This section will first formulate a CL-MRAC design for platooning, which serves as a basis for a novel adaptive platooning protocol, derived in the framework of C-MRAC.
\begin{figure*}[thbp]
    \centering
    \includegraphics[width=1\linewidth]{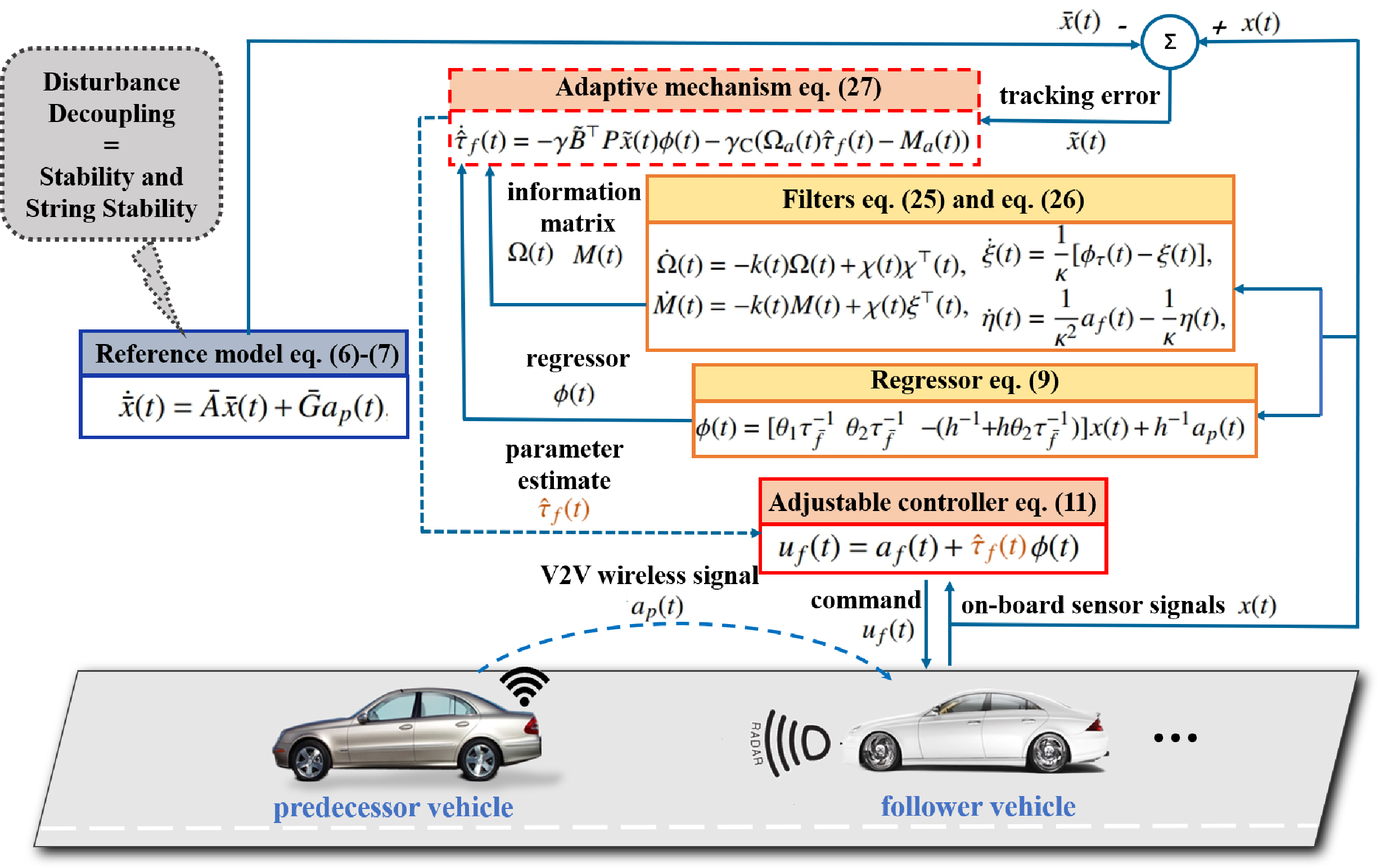}
    \caption{Framework of the proposed C-MRAC for platooning: by making use of disturbance decoupling as in Remark 1, we enforce stability and string stability in the reference model. 
    By asymptotically tracking the reference model, the predecessor-follower system asymptotically attains stability and string stability. 
    Asymptotic tracking is made possible by an adjustment mechanism to estimate the uncertain powertrain time constant: here, filtering avoids the need to measure acceleration derivative.\label{framework}}
\end{figure*}
\subsection{Concurrent Learning Design}
The principle of CL is using current data $\phi(x(t),a_p(t))$ and a set of recorded data $\phi(t_j) = \phi(x(t_j),a_p(t_j)),\ j\in \{1,2,\dots,k\}$ concurrently for adaptation to achieve convergence of the estimates to their true values.
The following results remind of CL-MRAC in\cite{c14}, but a different design of the adaptive law is used. 
Due to the different design, it is convenient to elaborate on its stability analysis.
\begin{theorem}\label{CL theorem}
Consider the predecessor-follower model (\ref{standard dyanamics}) with spacing error (\ref{spacing error}). 
The controller $u_f$ as in (\ref{actual controller}) with CL-MRAC adaptive law 
\begin{equation}\label{standard CL adaptive law}
    \begin{aligned}
        \dot{\hat{\tau}}_f(t)=&-\gamma \tilde{B}^\top P\tilde{x}(t)\phi(t)\\
        &-\gamma_{{\rm CL}}\sum_{j=1}^{k}\frac{\phi(t_j)}{\hat{\tau}_f(t_j)}(\hat{\tau}_f(t)\dot{a}_f(t_j)-\hat{\tau}_f(t_j)\phi(t_j)),
    \end{aligned}
\end{equation}
where $\gamma_{{\rm CL}}>0$ is another learning gain,
guarantees, for any unknown $\tau_p,\ \tau_f$ and any bounded $u_p(\cdot)$, that (\ref{x convergence}) holds;
in addition, given a set of recorded data $Z = [\phi(t_1),\dots,\phi(t_k)] \in \mathbb{R}^{1 \times k}$, if $\rank(Z)=1$, i.e. RE holds, then (\ref{tau convergence}) can be achieved.
\end{theorem}
\begin{proof}
    Consider the following positive definite and radially unbounded Lyapunov candidate function:
    \begin{equation}\label{Lyacan}
        V(\tilde{x},\tilde{\tau}_f)=\frac{1}{2}\tilde{x}^\top P \tilde{x}+\frac{\tilde{\tau}_f^2}{2\gamma\tau_f}.
    \end{equation}

    It is obvious that $V({\bm0},0)=0$ and $V(\tilde{x},\tilde{\tau}_f)>0,\ \forall(\tilde{x},\tilde{\tau}_f)\neq({\bm0},0)$.
    Let $\psi={[\tilde{x}^{\top}\ \tilde{\tau}_f]}^\top$.
    Then, from (\ref{Lyacan}) we have
    \begin{equation}
        \begin{aligned}\label{Lyacan inequality}
            \frac{1}{2}{\min}\{\lambda_{\min}(P),\frac{1}{\gamma\tau_f}\}{\Vert\psi\Vert}^2 &\leq V(\tilde{x},\tilde{\tau}_f)\\ 
            &\leq \frac{1}{2}{\max}\{\lambda_{\max}(P),\frac{1}{\gamma\tau_f}\}{\Vert\psi\Vert}^2.
        \end{aligned}
    \end{equation} 

    Next, from (\ref{actual tracking error}), one can derive that
    \begin{equation}\label{identity}
        \begin{aligned}
            \dot{a}_f(t)&=\theta_1\tau_{\bar{f}}^{-1}e_f(t)+\theta_2\tau_{\bar{f}}^{-1}\nu_f(t)-(h^{-1}+h\theta_2\tau_{\bar{f}}^{-1})a_f(t)\\
                        &\quad+h^{-1}a_p(t)+\tau_f^{-1}\tilde{\tau}(t)\phi(t)\\
                        &=\phi(t)+\tau_f^{-1}\hat{\tau}_f(t)\phi(t)-\phi(t)=\tau_f^{-1}\hat{\tau}_f(t)\phi(t),
        \end{aligned}
    \end{equation}
    which implies, for any of the recorded data,
    \begin{equation*}
        \frac{\phi(t_j)}{\tau_f}=\frac{\dot{a}_f(t_j)}{\hat{\tau}_f(t_j)},\ j\in\{1,\dots,k\}.
    \end{equation*}
    Note that 
    \begin{equation*}
        \begin{aligned}
            \tau_f^{-1}\tilde{\tau}_f(t)\phi(t_j)&=\tau_f^{-1}(\hat{\tau}_f(t)\phi(t_j)-\tau_f\phi(t_j))\\
                                                 &=\hat{\tau}_f(t)\frac{\dot{a}_f(t_j)}{\hat{\tau}_f(t_j)}-\phi(t_j),
        \end{aligned}
        \quad j\in\{1,\dots,k\}.
    \end{equation*}
    The adaptive law (\ref{standard CL adaptive law}) can be rewritten as
    \begin{equation}\label{standard CL law for proof}
        \dot{\hat{\tau}}_f(t)=-\gamma\tilde{B}^\top P\tilde{x}(t)\phi(t)-\gamma_{\rm CL}\tau_f^{-1}\sum_{j=1}^{k}{\phi^2(t_j)}\tilde{\tau}_f(t).
    \end{equation}

    Utilizing (\ref{actual tracking error}) and (\ref{standard CL law for proof}), the time derivative of (\ref{Lyacan}) along the trajectory can be obtained as
    \begin{equation}\label{derivative of CL Lyacan}
        \dot{V}(\tilde{x},\tilde{\tau}_f) = -\frac{1}{2}\tilde{x}^\top Q \tilde{x} - \gamma_{\rm CL}\frac{{\tilde{\tau}_f}^2}{\gamma\tau_f^2}\sum_{j=1}^{k}{\phi^2(t_j)}.
    \end{equation}
    Even in the absence of RE, (\ref{derivative of CL Lyacan}) is negative semidefinite, from which one can show convergence of $\tilde{x}$ using Barlalat's lemma as in standard MRAC. 
    If, in addition, $\phi(\cdot)$ satisfies RE, let $\Theta=\sum_{j=1}^{k}{\phi^2(t_j)}$ for which we have $\Theta>0$, leading to
    \begin{equation}
        \begin{aligned}\label{derivative of CL Lyacan inequality}
            \dot{V}(\tilde{x},\tilde{\tau}_f)&\le -\frac{1}{2}\lambda_{\min}(Q){\Vert \tilde{x} \Vert}^2-\gamma_{\rm CL}\frac{1}{\gamma\tau_f^2}\Theta{\tilde{\tau}_f}^2\\
            &\le -\frac{1}{2} \min\{\lambda_{\min}(Q),\ 2\gamma_{\rm CL}\frac{1}{\gamma\tau_f^2}\Theta\}{\Vert\psi\Vert}^2.
        \end{aligned}
    \end{equation}
    From (\ref{Lyacan inequality}) and (\ref{derivative of CL Lyacan inequality}), we have
    \begin{equation*}
        \dot{V}(\tilde{x},\tilde{\tau}_f)\leq - \frac{\min\{\lambda_{\min}(Q),\ 2\gamma_{\rm CL}{(\gamma\tau_f^2)}^{-1}\Theta\}}{\max\{\lambda_{\max}(P),\ {(\gamma\tau_f)}^{-1}\}} V(\tilde{x},\tilde{\tau}_f),
    \end{equation*}
    which is negative definite.
    By the Lyapunov stability theorem, it can be concluded that $(\tilde{x},\tilde{\tau}_f)=({\bm0},0)$ is globally exponentially stable.
    This ends the proof.
\end{proof}

A question arises about how many data points are needed in CL-MRAC to guarantee RE.
The condition $\rank([\phi(t_1), . . . , \phi(t_k)]) = 1$ simply requires to collect at least one non-zero data point, which relaxes considerably the PE condition. 
Based on (\ref{identity}), such data point exists if the derivative of the acceleration is non-zero at some time instant.
However, the adaptive law (\ref{standard CL adaptive law}) requires measurements of  $\dot{a}_f (t_j)$ and thus does not solve the control problem. 
In practice, it has been recognized that acceleration derivatives are hard to measure\cite{c2}.
In the following, an adaptive design solving the control problem is proposed.

\begin{figure}[t]
	\centering
	\subfloat[][Standard MRAC]{
		\centering                            
		\includegraphics[width=3.45in]{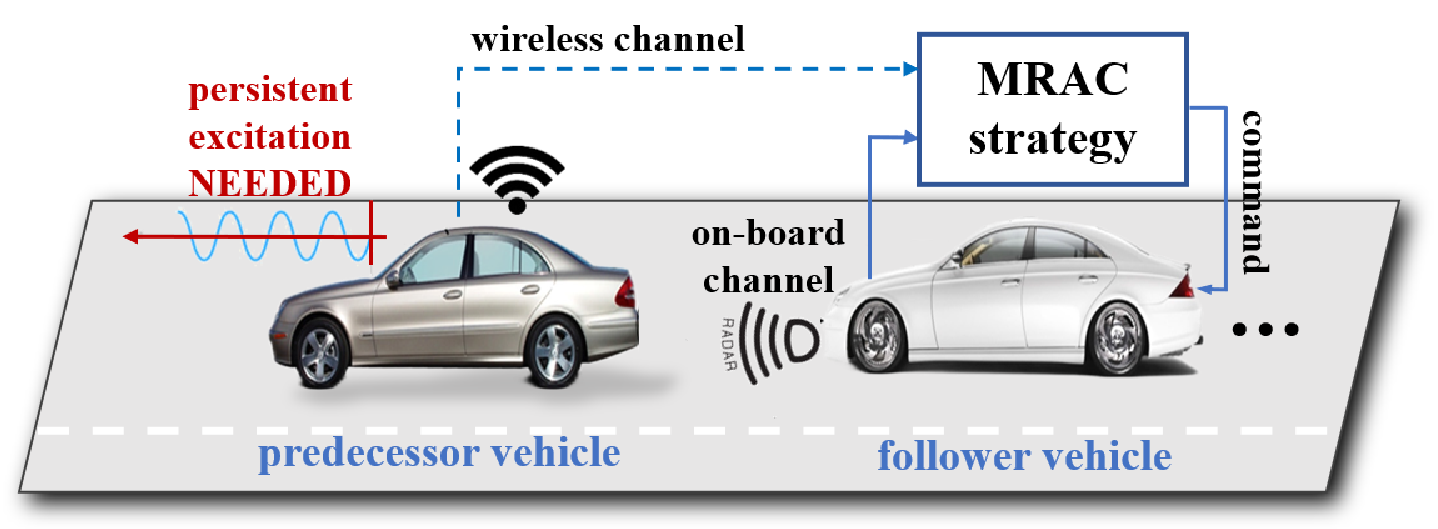}}
	
	\subfloat[][Concurrent Learning MRAC (CL-MRAC)]{
		\centering                            
		\includegraphics[width=3.45in]{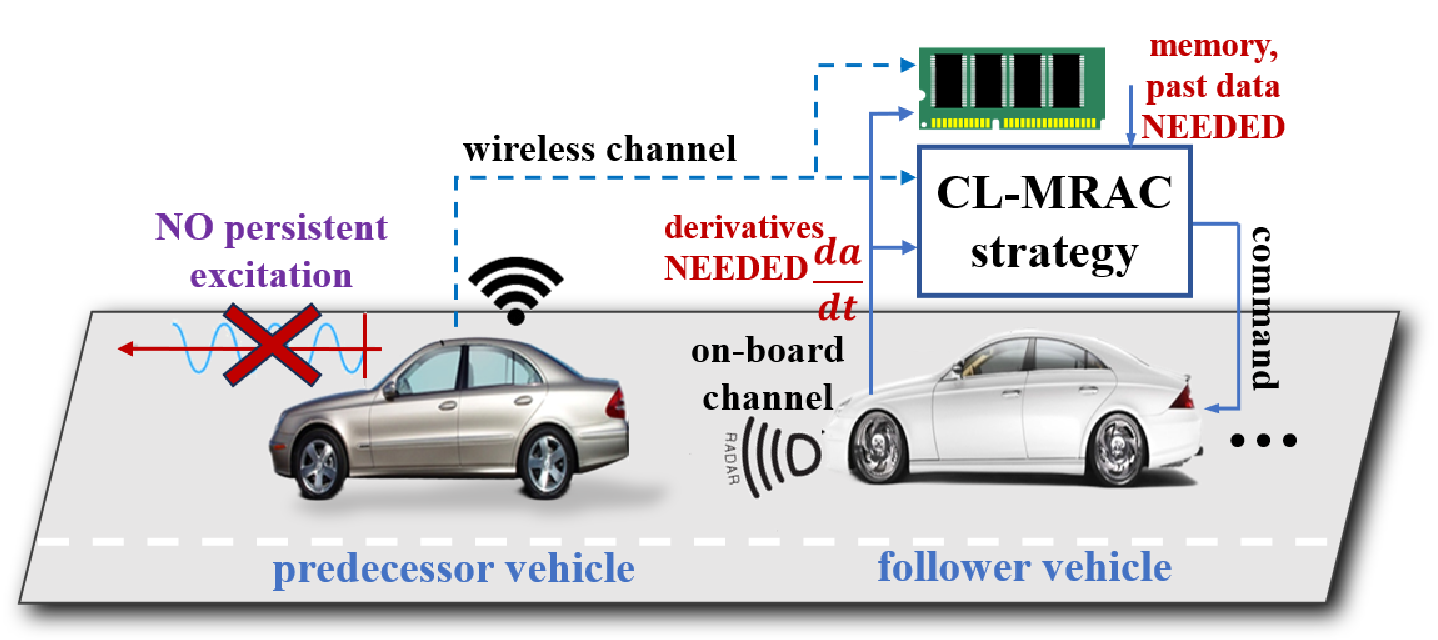}}
	
	\subfloat[][Composite MRAC (C-MRAC), proposed]{
		\centering                            
		\includegraphics[width=3.45in]{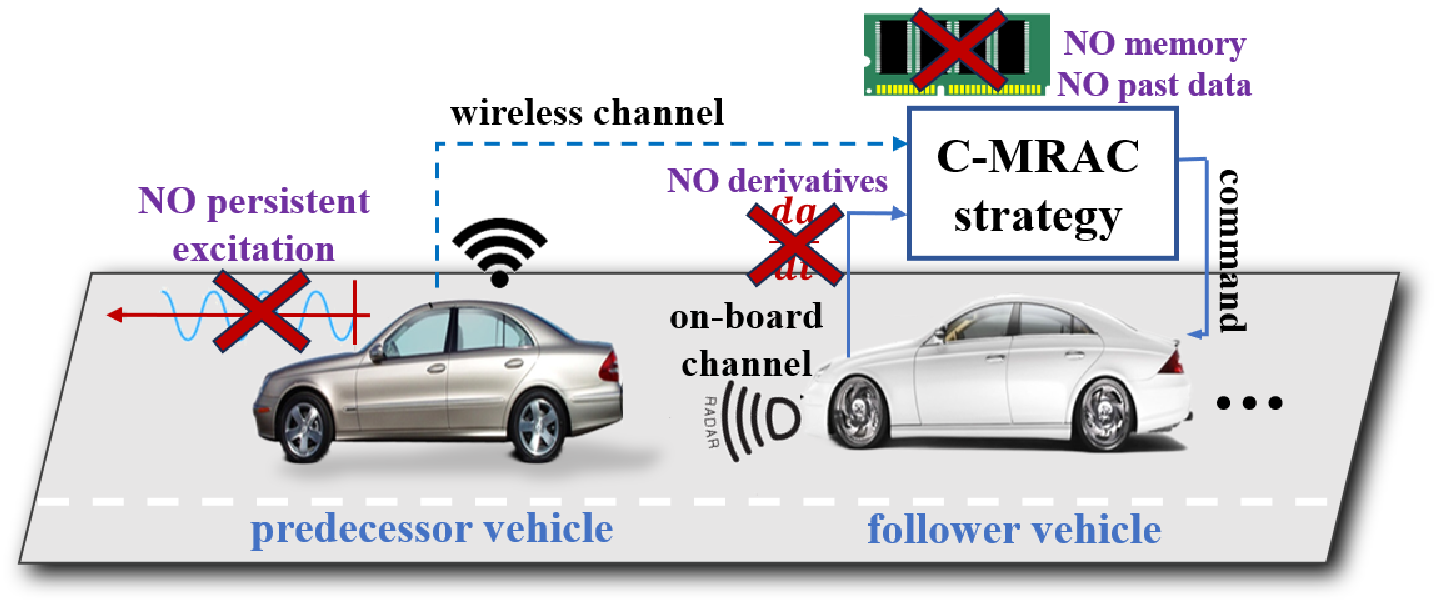}}    
	\caption{Proposed design as compared to the state of the art: (a) in standard MRAC, correct estimation of the uncertain powertrain time constant requires the precedessor to provide persistence of excitation; (b) concurrent learning MRAC removes the requirement for persistent excitation, at the price of introducing a memory with past data and additional measurements of acceleration derivative; (c) we propose a composite MRAC where no past data and no acceleration derivative are required.\label{contributions}}
\end{figure}

\subsection{Composite Design}
The first step to avoid measurements of state derivatives is to adopt a stable filter.
To this purpose, let us define $\phi_{\tau}(t)=\hat{\tau}_f(t) \phi(t)$ and multiply each side of (\ref{identity}) with a stable first-order filter $F(s)=\frac{1}{\kappa s+1}$ with $\kappa>0$.
We obtain
\begin{equation}
    \frac{\tau_f}{\kappa}(1-F(s))a_f(t)=F(s)\phi_{\tau}(t),\label{filtered identity s-domain}
\end{equation}
where we have used a hybrid time-frequency notation to indicate that a time signal is appropriately filtered. 
In the time domain, (\ref{filtered identity s-domain}) is equivalent to
\begin{equation}\label{xi generate}
    \xi(t)=\tau_f [\frac{1}{\kappa}a_f(t)-\eta(t)]\triangleq\tau_f\chi(t),
\end{equation}
where $\xi(t)$ and $\eta(t)$ are the results of the following filters
\begin{equation}\label{filtered system}
        \begin{aligned}
            \dot{\xi}(t)&=\frac{1}{\kappa}[\phi_{\tau}(t)-\xi(t)],\\
            \dot{\eta}(t)&=\frac{1}{\kappa^2}a_f(t)-\frac{1}{\kappa}\eta(t),
        \end{aligned}
\end{equation}
with initial conditions $\xi(t_0)=\eta(t_0)=0$.

Consistently with~\cite{c17}, let us design the information matrix $\Omega(t)$ and the auxiliary matrix $M(t)$ as
\begin{equation}\label{information matrix and auxiliary matrix}
        \begin{aligned}
            \dot{\Omega}(t) &= -k(t)\Omega(t)+\chi(t)\chi^\top(t),\\
            \dot{M}(t) &=-k(t)M(t)+\chi(t)\xi^\top(t),
        \end{aligned}
\end{equation}
with initial conditions $\Omega(t_0)=M(t_0)=0$, where $k(t)$ is a forgetting factor bounded by positive constants, i.e., $0<k_L\le k(t)\le k_U$.
A convenient design for $k(t)$ is
\begin{equation*}
        k(t)=k_L+(k_U-k_L)\tanh(\theta |\dot{\xi}(t)|),
\end{equation*}
where $\theta>0$ and $k_U>k_L>0$.
Note that, differently from $\dot{a}_f(t)$, the signal $\dot{\xi}(t)$ is not a state derivative to be measured, but a signal readily available as part of the filtering in (\ref{filtered system}).

Now, we are in a position to propose the C-MRAC platooning design.
\begin{theorem}
    Consider the predecessor-follower model (\ref{standard dyanamics}) with spacing error (\ref{spacing error}).
    The controller (\ref{actual controller}) with adaptive law
    \begin{equation}\label{Composite MRAC adaptive law}
        \dot{\hat{\tau}}_f(t)=-\gamma \tilde{B}^\top P\tilde{x}(t)\phi(t)-\gamma_{\rm C}(\Omega_a(t)\hat{\tau}_f(t)-M_a(t)),    
    \end{equation}
    where $\gamma_{\rm C}>0$ is a learning gain, and with
    \begin{equation*}
        \left\{
            \begin{aligned}
                t_a\triangleq \max({\arg\max}_{\substack{\ell\in[t_0,\ t]}}\Omega(\ell)),\\
                \Omega_a(t)\triangleq\Omega(t_a),\ M_a(t)\triangleq M(t_a),
            \end{aligned}
        \right.
    \end{equation*} 
    guarantees, for any unknown $\tau_p,\ \tau_f$ and any bounded $u_p(\cdot)$, that (\ref{x convergence}) holds.
    In addition, if there exist time instants $t_s>0$ and $t_e>t_s$, and $\alpha > 0$ such that $\int_{t_s}^{t_e} \chi^2(\ell)d\ell>\alpha>0$, i.e., FE holds over $[t_s,t_e]$, then (\ref{tau convergence}) can be achieved.
\end{theorem}
\begin{proof}
    Consider the Lyapunov candidate function (\ref{Lyacan}).
    Utilizing (\ref{actual tracking error}) and (\ref{Composite MRAC adaptive law}), the time derivative of (\ref{Lyacan}) along the trajectory can be obtained as
    \begin{equation}\label{derivative of C-MRAC Lyacan}
        \dot{V}(\tilde{x},\tilde{\tau}_f) = -\frac{1}{2}\tilde{x}^\top Q \tilde{x} - \Gamma_{\rm C}\frac{1}{\gamma\tau_f}\Omega_a{\tilde{\tau}_f}^2,
    \end{equation}
    where we have used the fact that, by solving (\ref{information matrix and auxiliary matrix}), one can get
    \begin{equation*}\label{solution of information matrix and auxiliary matrix}
        \begin{aligned}
            \Omega(t) &= \int_{t_0}^{t}(e^{\int_{\ell}^{t}-k(\zeta)d\zeta}\chi(\ell)\chi^\top(\ell))d\ell,\\
            M(t) &=\int_{t_0}^{t}(e^{\int_{\ell}^{t}-k(\zeta)d\zeta}\chi(\ell)\xi^\top(\ell))d\ell=\Omega(t)\tau_f,
        \end{aligned}
    \end{equation*}
    leading to
    \begin{equation*}\label{simplify the adaptive law}
        M_a(t)=\Omega_a(t)\tau_f.
    \end{equation*}
    Similar to Theorem~\ref{CL theorem}, (\ref{derivative of C-MRAC Lyacan}) is negative semidefinite, even in the absence of FE, which allows to obtain (\ref{x convergence}) using Barlalat's lemma as in standard MRAC. 
    If, in addition, $\chi(t)$ satisfies FE over $[t_s,t_e]$, then $\Omega_a>0$\cite{c17} and the following inequality holds
    \begin{equation}
        \begin{aligned}\label{derivative of C-MRAC Lyacan inequality}
            \dot{V}(\tilde{x},\tilde{\tau}_f)&\le -\frac{1}{2}\lambda_{\min}(Q){\Vert \tilde{x} \Vert}^2-\frac{1}{\gamma\tau_f}\Gamma_{\rm C}\Omega_a(t_e)\tilde{\tau}_f^2\\
            &\le -\frac{1}{2} \min\{\lambda_{\min}(Q),\ \frac{1}{\gamma\tau_f}2\gamma_{\rm C}\Omega_a(t_e)\}{\Vert\psi\Vert}^2.
        \end{aligned}
    \end{equation}

    From (\ref{Lyacan inequality}) and (\ref{derivative of C-MRAC Lyacan inequality}), we have
    \begin{equation*}
            \dot{V}(\tilde{x},\tilde{\tau}_f)\leq- \frac{\min\{\lambda_{\min}(Q),\ 2\Gamma_{\rm C}\Omega_a(t_e){(\gamma\tau_f)}^{-1}\}}{\max\{\lambda_{\max}(P),\ {(\gamma\tau_f)}^{-1}\}} V(\tilde{x},\tilde{\tau}_f),
    \end{equation*}
    which is negative definite.
    By the Lyapunov stability theorem, it can be concluded that $(\tilde{x},\tilde{\tau}_f)=({\bm0},0)$ is globally exponentially stable.
    This ends the proof. 
\end{proof}

\begin{figure*}[thbp]
    \centering
    \subfloat[][Non-adaptive control when $u_0$ guarantees PE.\label{ev non-adaptive PE}]{
        \centering                            
        \includegraphics[width=0.435\linewidth]{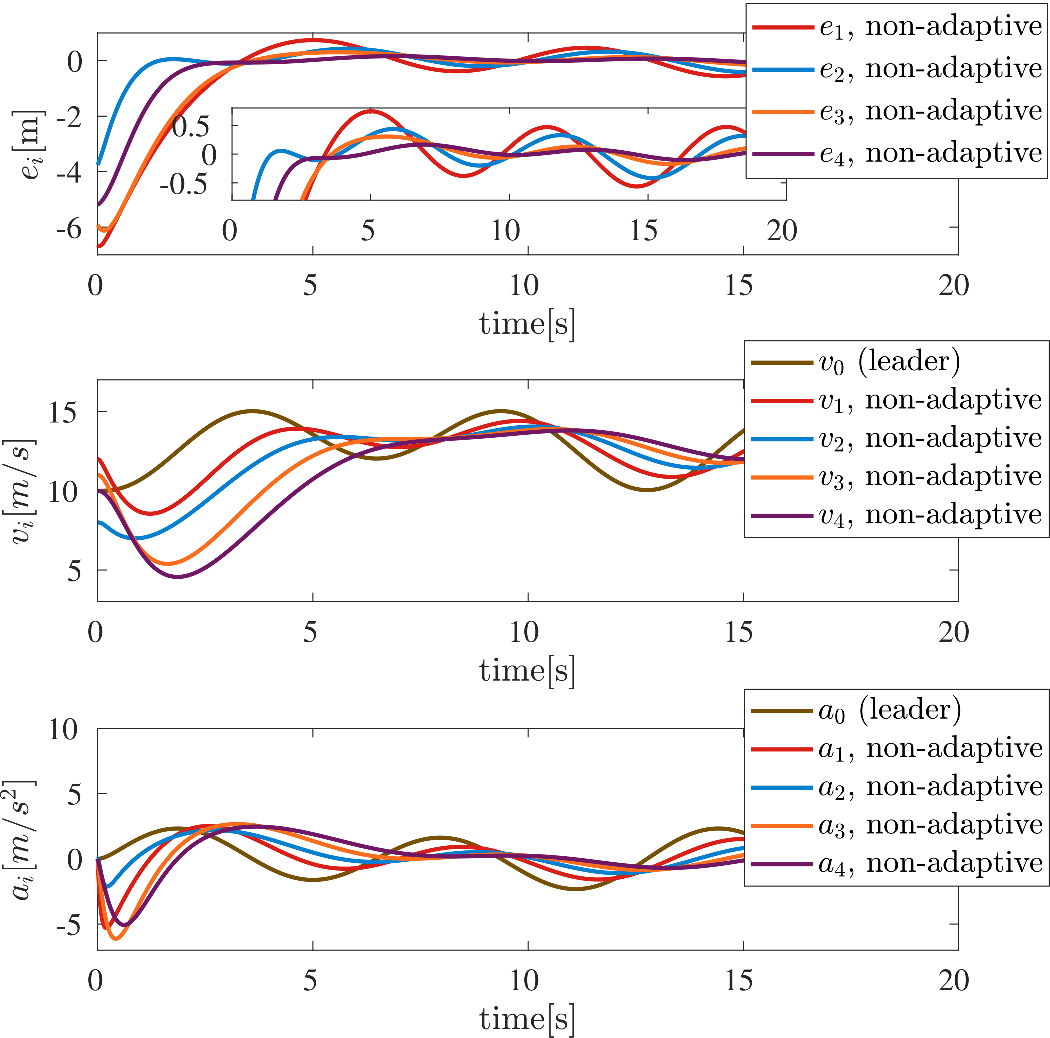}}   
    \hspace{4mm}
        \subfloat[][Non-adaptive control when $u_0$ does \emph{not} guarantee PE.\label{ev non-adaptive FE}]{
        \centering                            
        \includegraphics[width=0.435\linewidth]{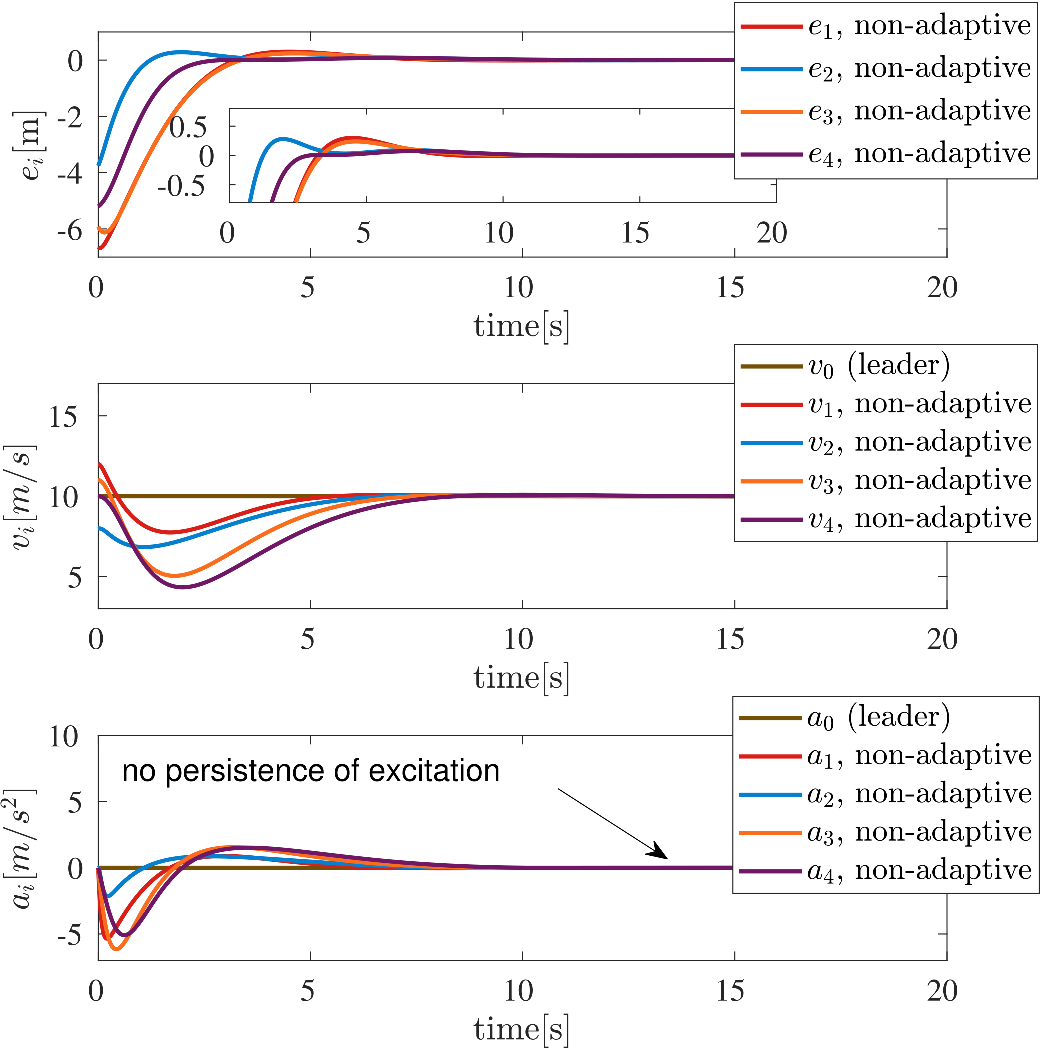}}
    \caption{Spacing errors, velocities and accelerations with non-adaptive controller (\ref{reference model input}) without knowledge of the actual $\tau_i$. Note that the spacing errors in (a) do not converge to zero as an effect of the non-zero acceleration of the leading vehicle.\label{ev_non-adaptive}}
\end{figure*}

\begin{figure*}
    \centering
    \subfloat[][MRAC and C-MRAC when $u_0$ guarantees PE.]{
        \centering                            
        \includegraphics[width=0.435\linewidth]{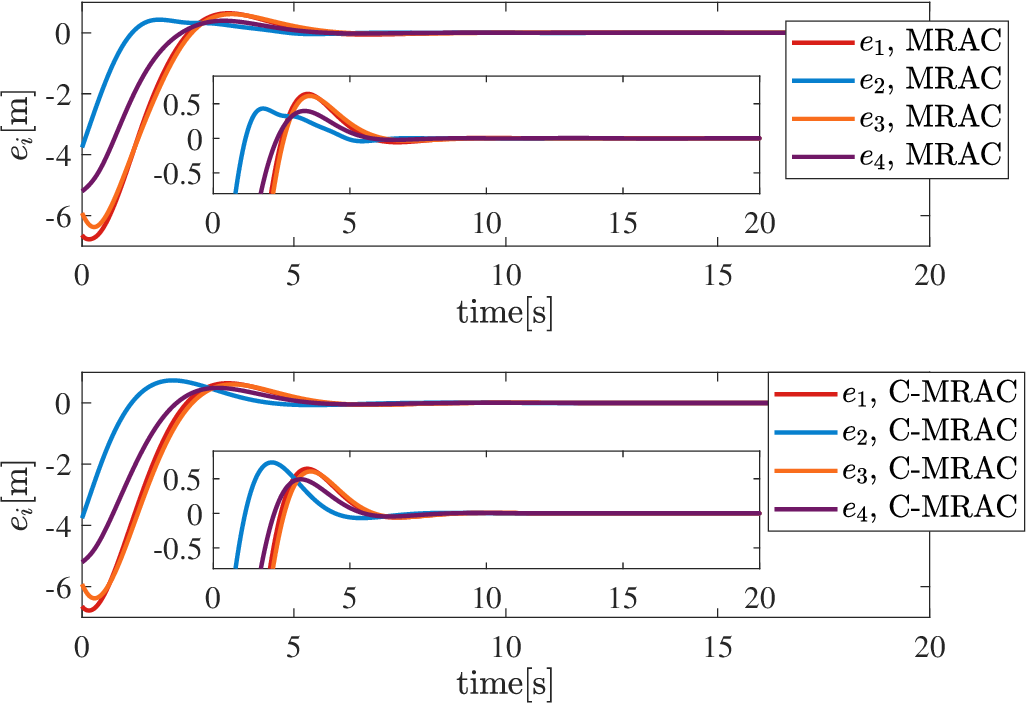}}   
    \hspace{4mm}
    \subfloat[][MRAC and C-MRAC  when $u_0$ does \emph{not} guarantee PE.]{
        \centering                            
        \includegraphics[width=0.425\linewidth]{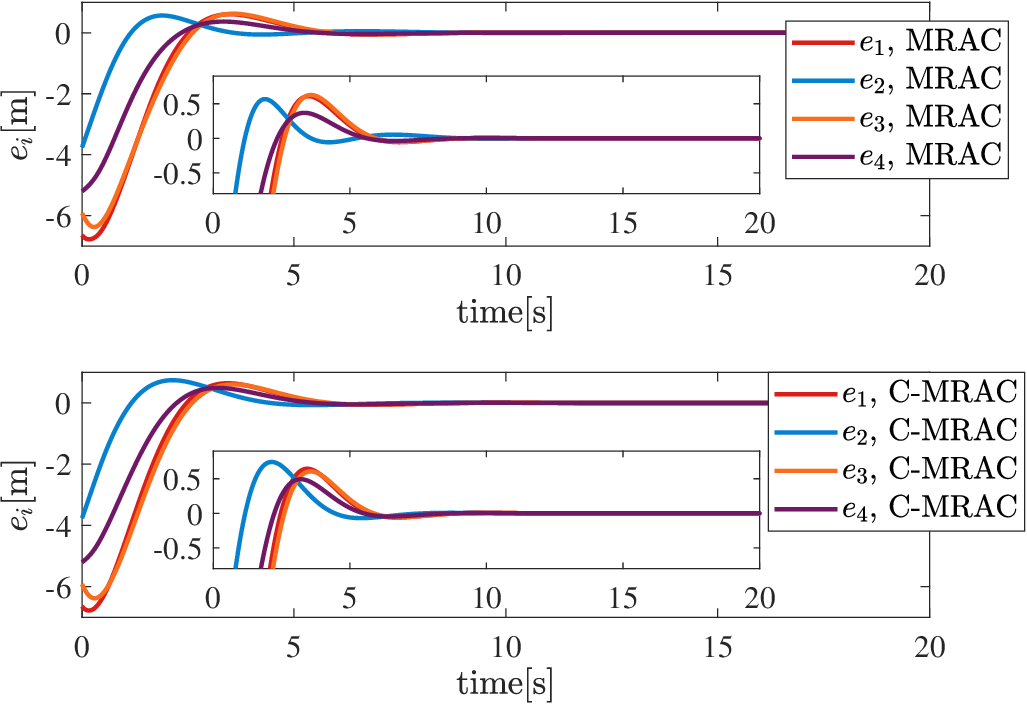}}
    \captionsetup{justification=centering}
    \caption{Spacing errors with adaptive controllers (MRAC (\ref{standard MRAC adaptive law}) and C-MRAC (\ref{Composite MRAC adaptive law})) without knowledge of the actual $\tau_i$. No matter if the leading vehicle has non-zero or zero acceleration, the adaptation allows the spacing errors to converge to zero.\label{e_adaptive}}      
\end{figure*}

\begin{figure*}
    \centering
    \subfloat[][MRAC (\ref{standard MRAC adaptive law}) when $u_0$ guarantees PE.]{
        \centering                            
        \includegraphics[height=5.1cm,width=0.4\linewidth]{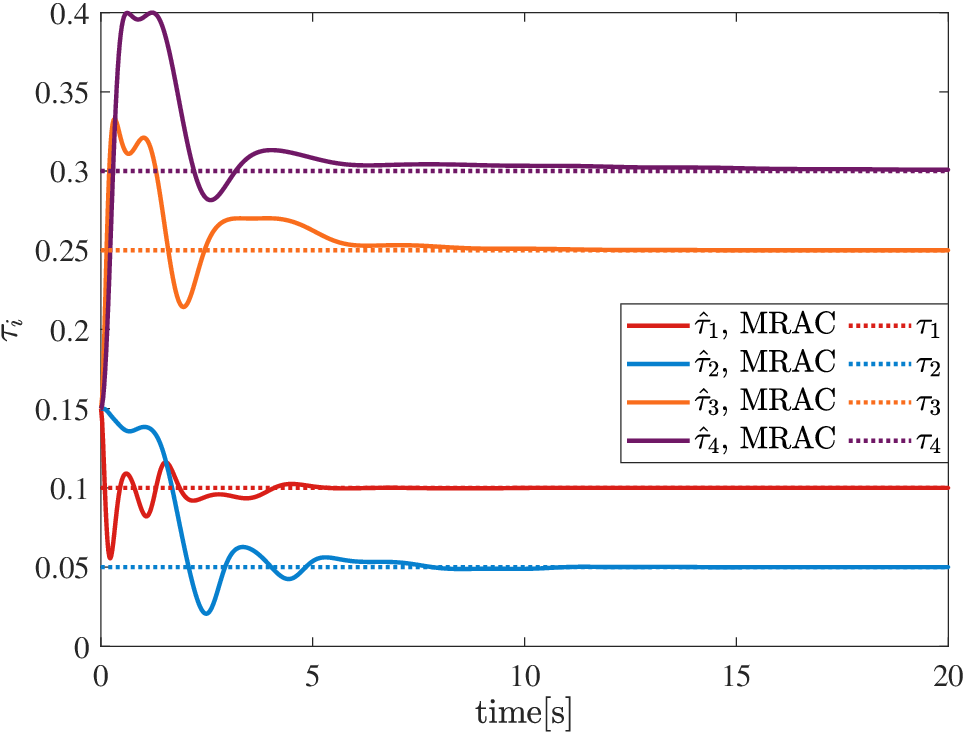}}   
    \hspace{12mm}
    \subfloat[][MRAC (\ref{standard MRAC adaptive law}) when $u_0$ does \emph{not} guarantee PE.]{
        \centering                            
        \includegraphics[height=5.1cm,width=0.4\linewidth]{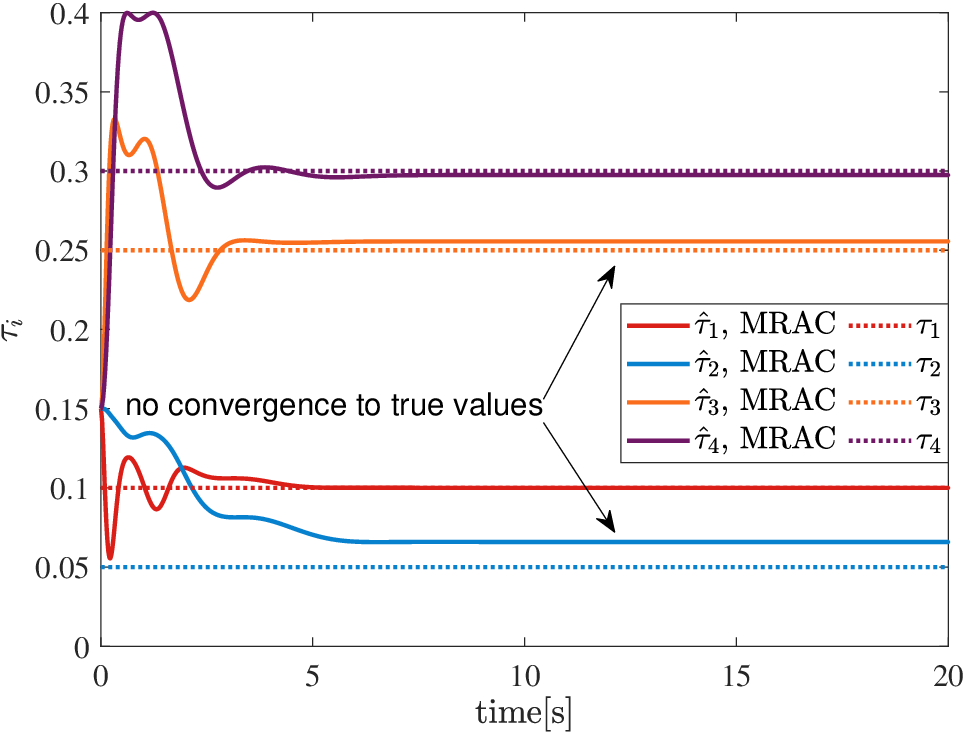}}
         
    \subfloat[][CL-MRAC (\ref{standard CL adaptive law}) when $u_0$ guarantees PE.]{
        \centering                            
        \includegraphics[height=5.1cm,width=0.4\linewidth]{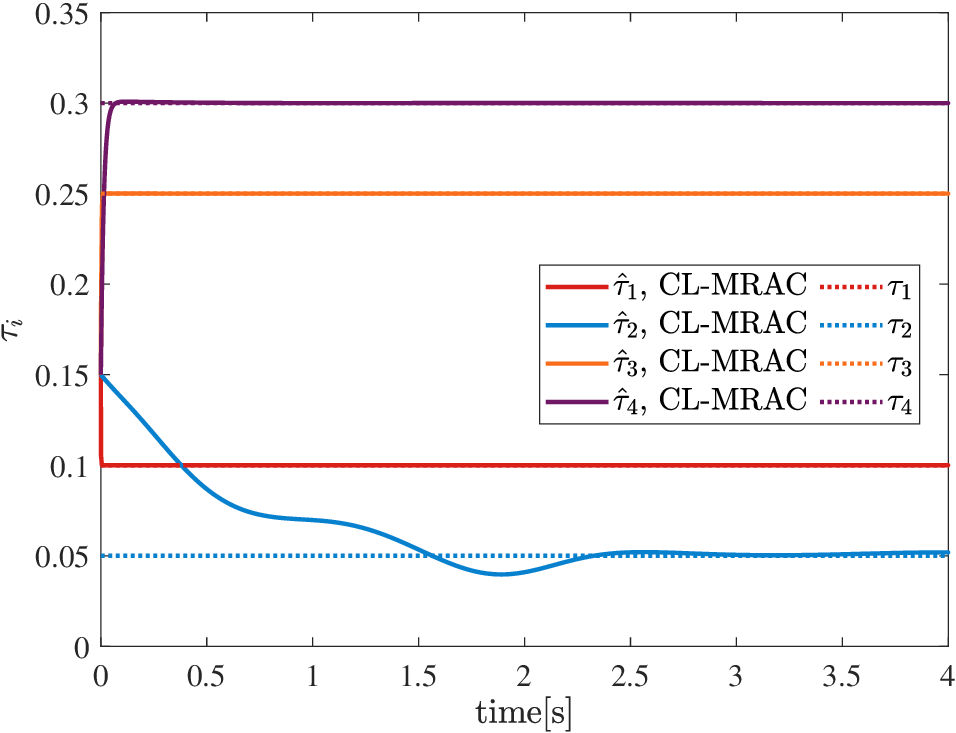}}   
    \hspace{12mm}
    \subfloat[][CL-MRAC (\ref{standard CL adaptive law}) when $u_0$ does \emph{not} guarantee PE.]{
        \centering                            
        \includegraphics[height=5.1cm,width=0.4\linewidth]{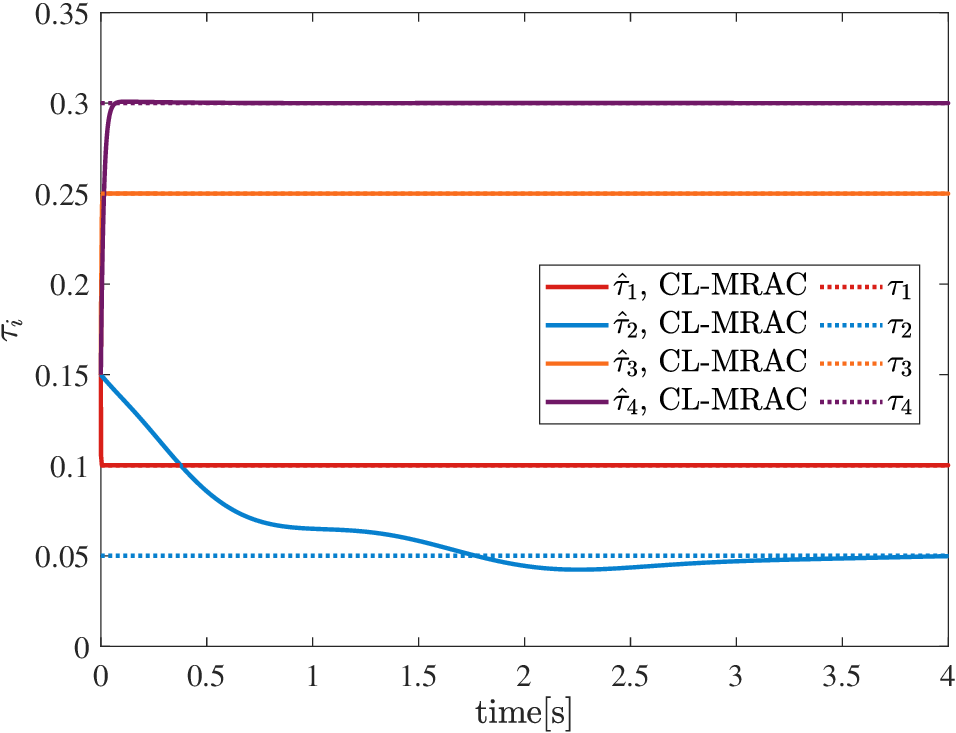}}
            
    \subfloat[][ICL-MRAC (\ref{ICL adaptive law}) when $u_0$ guarantees PE.]{
        \centering                            
        \includegraphics[height=5.1cm,width=0.4\linewidth]{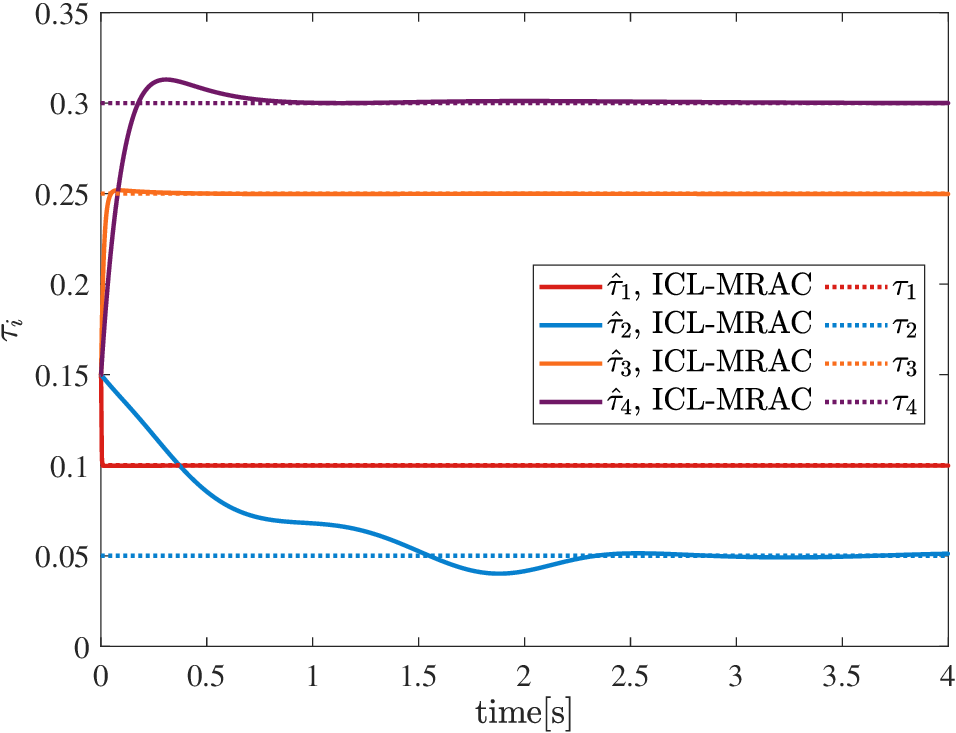}}   
    \hspace{12mm}
    \subfloat[][ICL-MRAC (\ref{ICL adaptive law}) when $u_0$ does \emph{not} guarantee PE.]{
        \centering                            
        \includegraphics[height=5.1cm,width=0.4\linewidth]{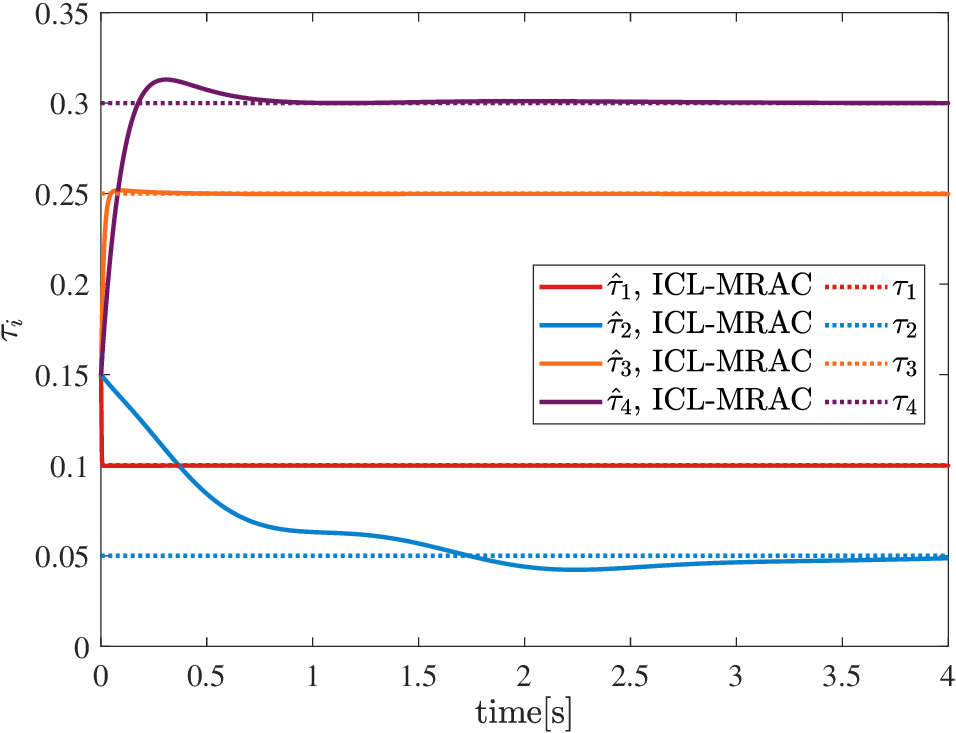}}

    \subfloat[][C-MRAC (\ref{Composite MRAC adaptive law}) when $u_0$ guarantees PE.]{
        \centering                            
        \includegraphics[height=5.1cm,width=0.4\linewidth]{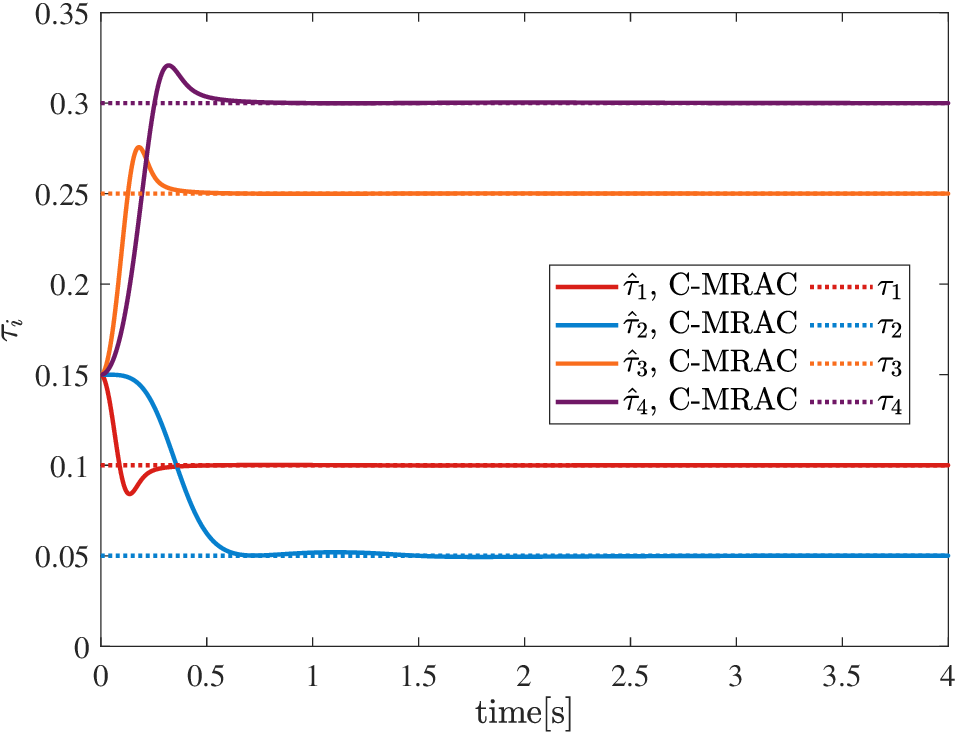}}   
    \hspace{12mm}
    \subfloat[][C-MRAC (\ref{Composite MRAC adaptive law}) when $u_0$ does \emph{not} guarantee PE.]{
        \centering                            
        \includegraphics[height=5.1cm,width=0.4\linewidth]{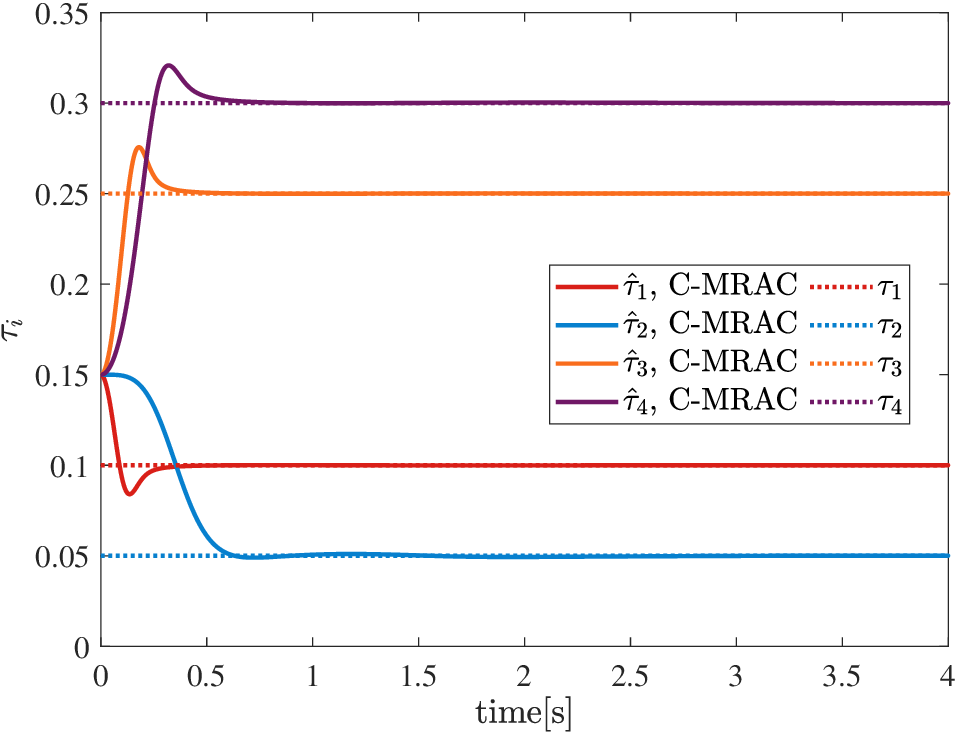}}
            
        \captionsetup{justification=centering}
        \caption{Estimates with different adaptive controllers without knowledge of the actual $\tau_i$.\label{tau_adaptive}}      
    \end{figure*}
    
    \begin{figure*}
    \centering
    \subfloat[][Spacing errors with MRAC (\ref{standard MRAC adaptive law}) and C-MRAC (\ref{Composite MRAC adaptive law}).\label{hybrid spacing error}]{
        \centering                            
        \includegraphics[width=0.435\linewidth]{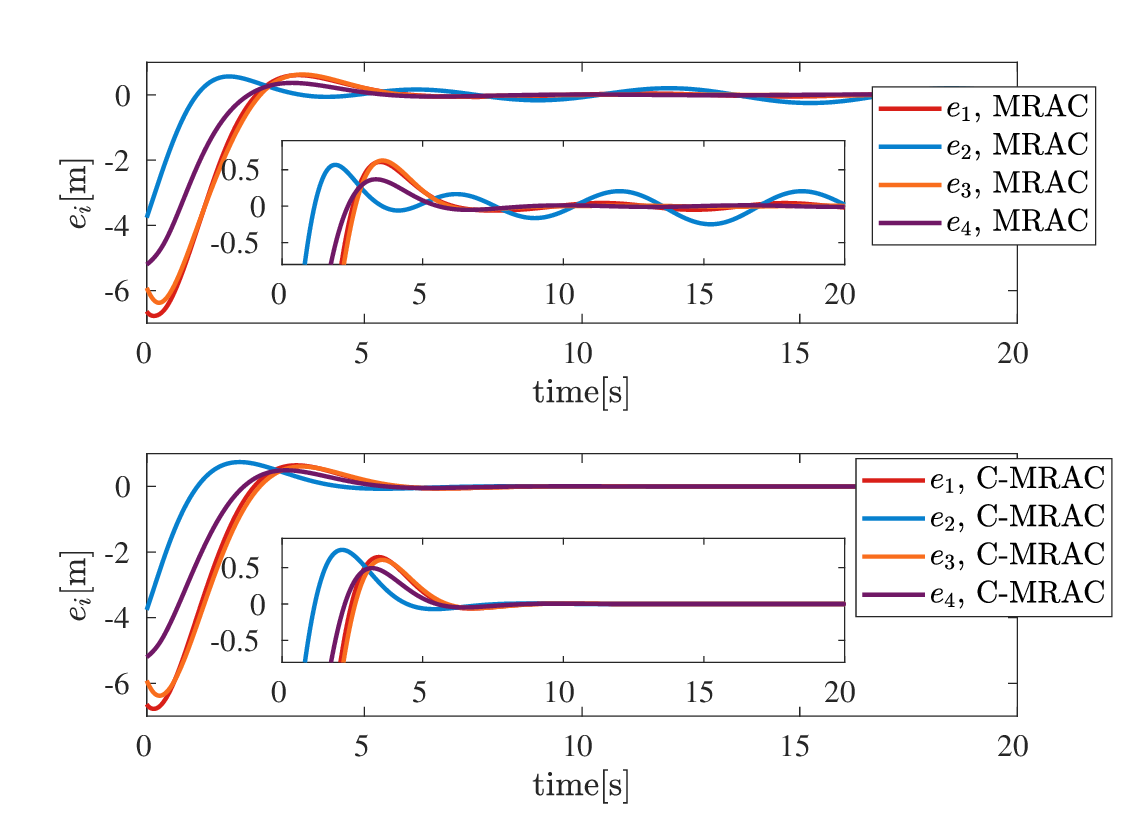}}   
    \hspace{4mm}
    \subfloat[][Estimates with MRAC (\ref{standard MRAC adaptive law}) and C-MRAC (\ref{Composite MRAC adaptive law}).\label{hybrid estimates}]{
        \centering                            
        \includegraphics[width=0.425\linewidth]{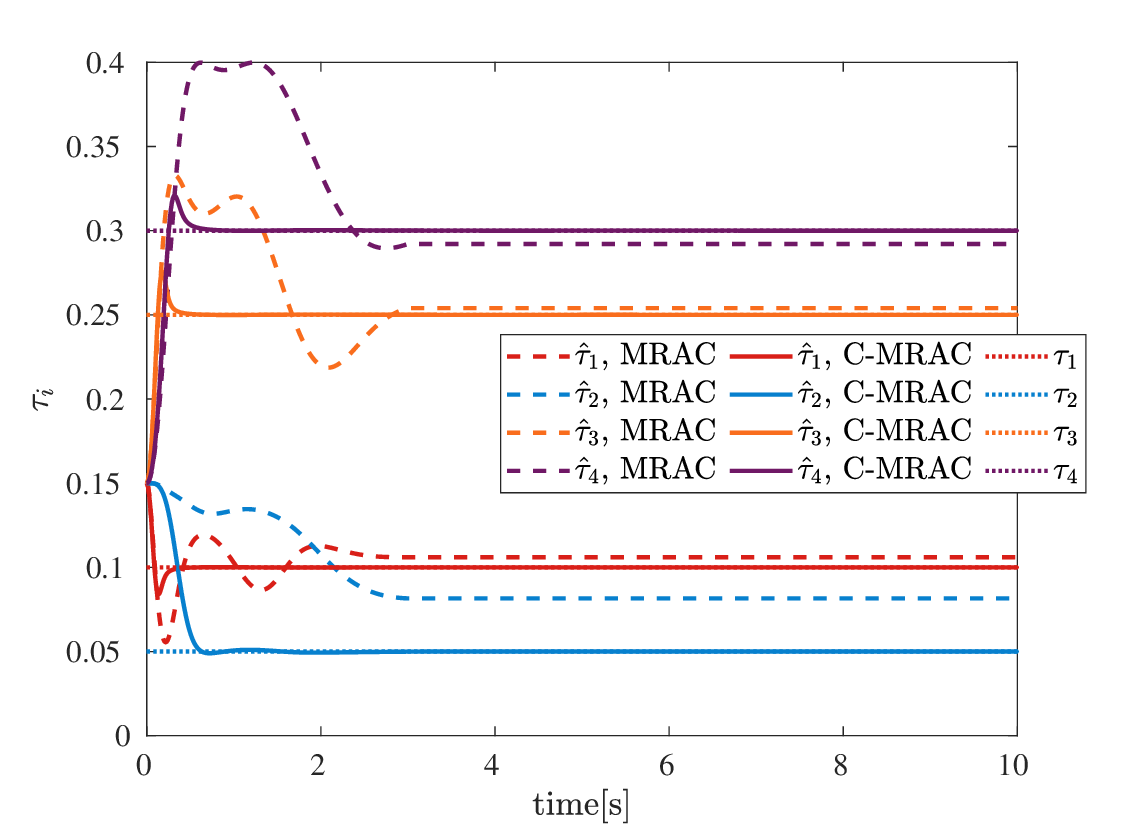}}
    \captionsetup{justification=centering}
    \caption{Spacing errors and estimates with MRAC (\ref{standard MRAC adaptive law}) and C-MRAC (\ref{Composite MRAC adaptive law}) when $u_0$ has piecewise acceleration profile and estimate  updates are stopped at $t=3$s.
    (a) Spacing errors with MRAC do not converge to zero; (b) MRAC fails to correctly estimate the powertrain time constants in 3 seconds. 
    On the other hand, C-MRAC correctly estimates in around 1 second, despite the input $u_0$ does \emph{not} guarantee PE.\label{hybrid}}      
\end{figure*}

    \begin{figure*}
        \centering
        \subfloat[][Non-adaptive control and C-MRAC.\label{carsim e PE}]{
            \centering                            
            \includegraphics[width=0.48\linewidth]{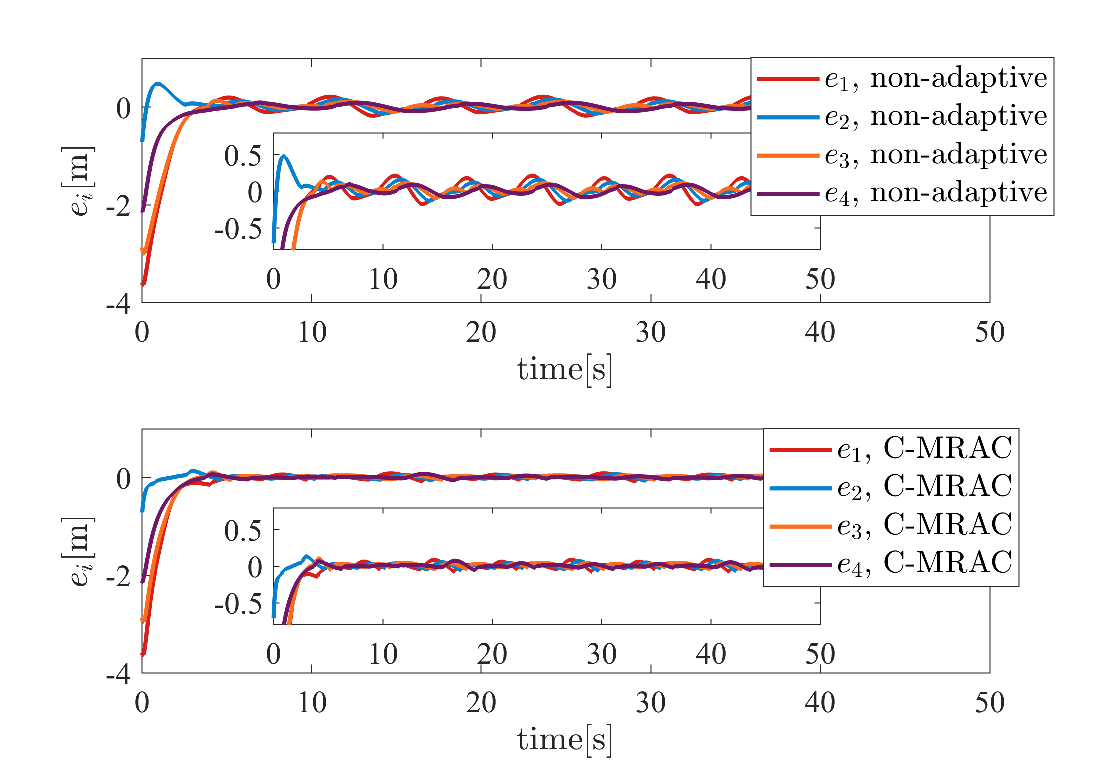}}   
        \hspace{2mm}
        \subfloat[][CarSim interface for vehicle 2.\label{carsim interface}]{
            \centering                            
            \includegraphics[width=0.48\linewidth]{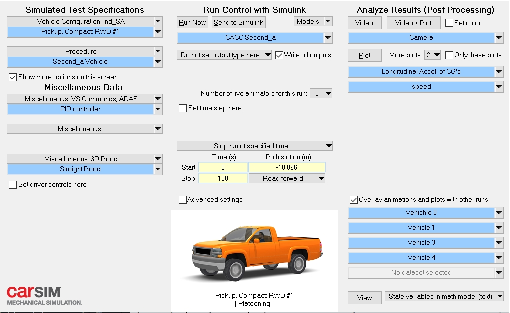}}       
        \captionsetup{justification=centering}
        \caption{CarSim test: spacing errors with non-adaptive controller (\ref{reference model input}) and adaptive controller C-MRAC (\ref{Composite MRAC adaptive law}). 
        Even without knowledge of the vehicle dynamics simulated in CarSim, the proposed adaptive method allows the spacing errors to converge close to zero (convergence is not asymptotic because the CarSim vehicle dynamics are more complex than the linear dynamics (\ref{standard dyanamics}) used for control design).\label{carsim e}}      
    \end{figure*}
    
    \begin{figure*}
        \centering
        \begin{minipage}[b]{.46\linewidth}
            \centering
            \subfloat[Initial condition ($t=0$s)\label{carsim initial condition}]{\includegraphics[height=1.8cm,width=1\linewidth]{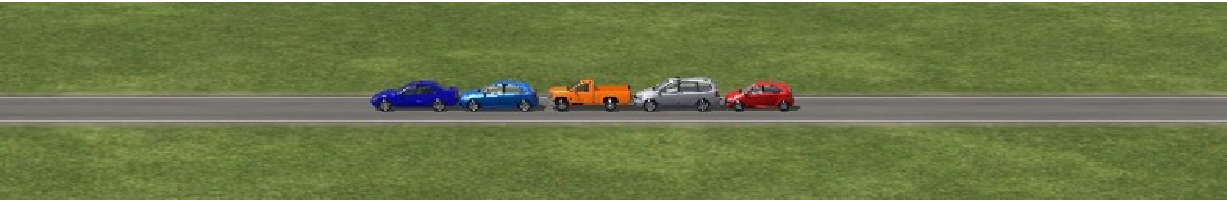}}
            \vspace{5mm}
            \subfloat[High velocity ($t=19$s)\label{carsim high speed}]{\includegraphics[height=1.8cm,width=1\linewidth]{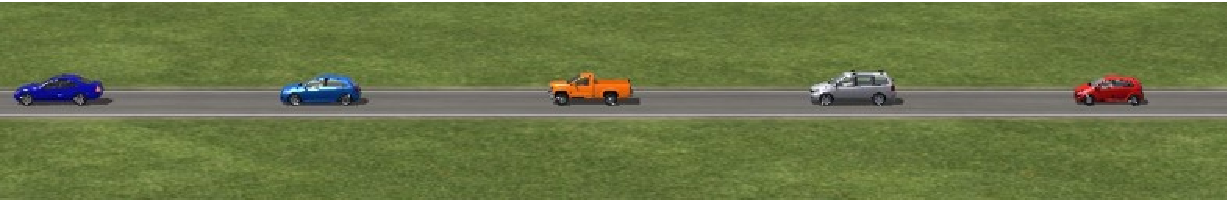}}
            \vspace{5mm}
            \subfloat[Low velocity ($t=34$s)\label{carsim low speed}]{\includegraphics[height=1.8cm,width=1\linewidth]{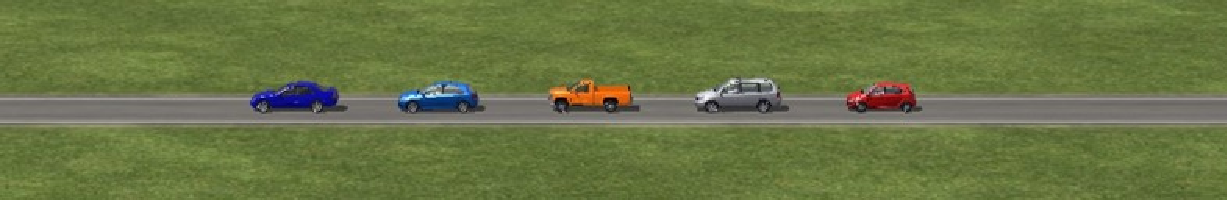}}
        \end{minipage} 
        \medskip
        \begin{minipage}[b]{.46\linewidth}
            \centering
            \subfloat[The longitudinal acceleration and velocity.\label{carsim a_v}]{\includegraphics[width=.95\linewidth]{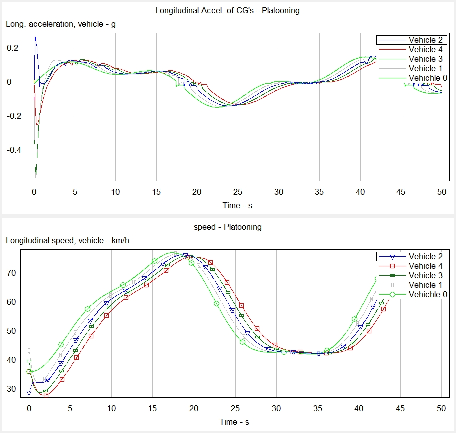}}
        \end{minipage}
        \captionsetup{justification=centering}
        \caption{CarSim visualization of the platoon at different time instants, of the accelerations and velocities of the vehicles. 
        In line with the time headway policy (\ref{spacing error}), the spacing grows as the velocity grows and decreases as the velocity decreases.\label{carsim platoon}}
    \end{figure*}
It can be noted that, based on (\ref{filtered identity s-domain})-(\ref{filtered system}), the FE condition on $\chi(t)$ is satisfied if the derivative of the acceleration is non-zero over a finite time interval, which is often the case in practice, even in situations where the platoon converges to constant velocity. 
This is consistent with the fact that the CL-MRAC in Theorem~\ref{CL theorem} requires only one time instant where the derivative of the acceleration is non-zero. 
However, differently from (\ref{standard CL adaptive law}), the adaptive law (\ref{Composite MRAC adaptive law}) does not require any acceleration derivative: it uses the same measurements used by the control law (\ref{actual controller}).
The framework of the proposed platooning design is shown in Fig.~\ref{framework}. 
\begin{remark}
The adaptive law (\ref{Composite MRAC adaptive law}) differs from the state-of-the-art composite adaptive control~\cite{c17}, which cannot handle the error dynamics (\ref{actual tracking error}) due to the presence of the uncertain $\tau_f$ in the input matrix. 
As a matter of fact,~\cite{c17} assumes the input matrix to be known. 
The proposed C-MRAC thus extends the state of the art.
\end{remark}

\subsection{Comparison with Existing Adaptive Schemes}
This section elaborates on the key differences of the proposed design with respect to state-of-the-art designs. 
Such designs will also be used in Section \Rmnum{4} for numerical comparisons.

\subsubsection{Adaptive Platooning with standard MRAC}
The adaptive law (\ref{standard MRAC adaptive law}) requires PE condition of $\phi(t)$ for the convergence of $\hat{\tau}_f(t)$ to $\tau_f$.
However, PE condition is rather restrictive in practice, e.g., a platoon converging to constant velocity would converge to zero acceleration and fail to meet PE.

\subsubsection{Adaptive Platooning with CL-MRAC}
With the adaptive law (\ref{standard CL adaptive law}), the PE condition can be relaxed to RE condition.
To achieve this property, CL-MRAC uses a set of previously collected data. 
The advantage is that only one non-zero data sample is enough to guarantee RE. 
However, in order to collect such data samples, measurements of acceleration derivative $\dot{a}_f(t_j)$ is required, which is the main drawback of CL-MRAC.

\subsubsection{Adaptive Platooning with ICL-MRAC}
If ICL in\cite{c15} is applied to adaptive platooning, a suitable design is as follows:
\begin{equation}\label{ICL adaptive law}
    \begin{aligned}
        \dot{\hat{\tau}}_f(t)&=-\gamma\tilde{B}^\top P\tilde{x}(t)\phi(t)\\
        &\quad-\gamma_{\rm ICL}\sum_{j=1}^{k} \mathcal{A}_j(\mathcal{A}_j\hat{\tau}_f(t)-\int_{t_j-\Delta t}^{t_j}\hat{\tau}_f(\ell)\phi(\ell)d\ell),
    \end{aligned}
\end{equation}
where $\gamma_{\rm ICL}>0$ is another learning rate, $\Delta t>0$ denotes the size of the window of integration and $\mathcal{A}_j=a_f(t_j)-a_f(t_j-\Delta t)$.
The proof is omitted but follows along similar steps as Theorem~\ref{CL theorem} and the considerations in~\cite{c15}.
The adaptive law (\ref{ICL adaptive law}) replaces the time derivative with an integral over a time window. 
By doing this, the PE condition can be relaxed to IRE condition, requiring only one sample where the difference of acceleration, i.e., $a_f(t)-a_f(t-\Delta t)$, is non-zero.
However, different from (\ref{Composite MRAC adaptive law}), the adaptive law (\ref{ICL adaptive law}) still relies on previously collected data samples.

We summarize some representative comparison aspects of different designs in Fig~\ref{contributions}.

\section{Numerical validation}
To validate the theoretical analysis, we make use of a platoon with five vehicles (one leader indexed as $0$ and four following vehicles indexed as $1,\ 2,\ 3,\ 4$), with initial conditions shown in Table~\ref{table1}.
In the table, $s_i(0),\ i\in\{0,1,\dots,4\}$ refers to the initial position of each vehicle with respect to the leading vehicle.
We first provide Matlab-based simulations and then CarSim-based simulations, to validate the robustness of the proposed protocol to more complex vehicle dynamics.
\begin{table}[h]
    \caption{Initial conditions}
   \label{table1}
    \begin{center}
    \begin{tabular}{|c|c|c|c|}
    \hline
    $i$ & $s_i(0)\ [m]$ & $v_i(0)\ [m/s]$ & $a_i(0)\ [m/s^2]$\\
    \hline
    $0$  & 0 & 10 & 0\\
    \hline
    $1$ & -2 & 12 & 0\\
    \hline
    $2$ & -4 & 8 & 0\\
    \hline
    $3$  & -6 & 11 & 0\\
    \hline
    $4$ & -8 & 10 & 0\\
    \hline
    \end{tabular}
    \end{center}
\end{table}

We first take vehicle dynamics as in (\ref{standard dyanamics}), where the time constants of each vehicle are taken as: $\tau_0=0.2,\ \tau_1=0.1,\ \tau_2=0.05,\ \tau_3=0.25,\ \tau_4=0.3$.
These values are chosen based on the literature reporting values of the time constants in the range we consider\cite{c1,c2,c3,c4,c5}.
The design parameters of the controllers are taken as $\theta_1=1,\ \theta_2=1,\ h=0.72,\ Q=0.69{\bf I},\ k_L=0,\ k_U=1,\ \theta=0.1, \kappa=0.25,\ \gamma=0.35,\ \gamma_{\rm C}=10,\ \gamma_{\rm CL}=0.3,\ \gamma_{\rm ICL}=68/\Delta t,\ \Delta t = 10^{-4}$ and $\tau_{\bar{f}} = 0.5$ as nominal powertrain time constant for the reference model.  

We consider five types of controllers:
\begin{enumerate}
    \item a non-adaptive controller as in (\ref{reference model input}) with incorrect knowledge of $\tau_i,\ i\in\{0,1,2,3,4,5\}$, which we select as $0.15$, 
    \item an adaptive controller (\ref{actual controller}) with MRAC design (\ref{standard MRAC adaptive law}),
    \item an adaptive controller (\ref{actual controller}) with CL-MRAC design (\ref{standard CL adaptive law}),
    \item an adaptive controller (\ref{actual controller}) with ICL-MRAC design (\ref{ICL adaptive law}),
    \item an adaptive controller (\ref{actual controller}) with C-MRAC design (\ref{Composite MRAC adaptive law}).
\end{enumerate}
The five scenarios are considered both in the presence of PE and in the absence of PE. 
To simulate PE, we let the leader vehicle take a sinusoidal acceleration, i.e., $u_0(t)=2\sin(t)+0.5\sin(0.5t),\ t\ge0$. 
To simulate absence PE, we let the leader proceed at constant velocity with $u_0(t)=0,\ t\ge0$. 
Then, to further illustrate the importance of convergence to the correct powertrain time constants, we let the leader vehicle take a piecewise acceleration, i.e., 
\begin{equation*}
u_0(t)=
\left\{
    \begin{aligned}
&0,\qquad\qquad\quad 0\le t\le3,\\
2\sin(t)&+0.5\sin(0.5t),\ t>3,    
\end{aligned}
\right.    
\end{equation*}
and stop updating the estimates at $t=3$s.
Let us remark that even when PE is absent, still RE, IRE and FE hold thanks to the non-zero initial conditions (for spacing error and relative velocity) that make the acceleration (or its time derivative) non-zero over some transient.

Fig.~\ref{ev_non-adaptive} shows spacing errors and accelerations of each vehicle with non-adaptive controller.
Fig.~\ref{e_adaptive} shows the spacing errors  with adaptive controllers.
In particular, Fig.~\ref{ev_non-adaptive}(a) shows that with incorrect knowledge of $\tau_i,\ i\in\{1,2,3,4\}$, disturbance decoupling is not achieved by a non-adaptive control (note that the spacing errors oscillate without converging to zero as an effect of the oscillating acceleration of the leading vehicle). 
The spacing errors in Fig.~\ref{ev_non-adaptive}(b) converge to zero even in the non-adaptive case because the leading vehicle converges to zero acceleration.
Concerning the adaptive controllers, Fig.~\ref{e_adaptive} shows that, both MRAC and C-MRAC achieve disturbance decoupling: the spacing errors converge to zero, no matter if the leading vehicle has non-zero or zero acceleration.
It can also be seen that all accelerations $a_i,\ i\in\{1,2,3,4\}$ satisfy PE in Fig.~\ref{ev_non-adaptive}(a) but do not satisfy PE in Fig.~\ref{ev_non-adaptive}(b).  

Fig.~\ref{tau_adaptive} shows the estimates of $\tau_i,\ i\in\{1,2,3,4\}$ with all adaptive controllers.
When PE is satisfied, all adaptive controllers can guarantee the correct estimation of powertrain time constants.
On the contrary, in the absence of PE, correct estimation is not achieved for standard MRAC.
Correct convergence to the true value is achieved in CL-MRAC, ICL-MRAC, C-MRAC because RE, IRE and FE all hold thanks to the non-zero acceleration over some transient.

Fig.~\ref{hybrid} shows the spacing errors of each vehicle and the estimates of $\tau_i,\ i\in\{1,2,3,4\}$ for standard MRAC and proposed C-MRAC.
It can be seen from Fig.~\ref{hybrid}(b) that C-MRAC correctly estimates the powertrain time constants within 1 second, despite the absence of PE, while MRAC cannot guarantee the correct estimation after 3 seconds.
After $t=3$s, when the update of the estimates is stopped, Fig.~\ref{hybrid}(a) shows that MRAC fails achieve disturbance decoupling as a result of the wrong estimation of the powertrain time constants.
In other words, Fig.~\ref{tau_adaptive} validates the well-known fact that standard MRAC suffers from lack of robustness due to convergence of the estimates to incorrect values\cite{c29,c30}.

\subsection{CarSim validation}
To test the robustness of the proposed ideas beyond the vehicle dynamics (\ref{standard dyanamics}), we simulate realistic vehicle dynamics by making use of CarSim.
The vehicle models used in CarSim cover high fidelity modules for handling, ride, vehicle stability, acceleration and braking  dynamics. 
Among the various modules, we make use of a low-level PID controller to make the actual acceleration follow the desired acceleration by controlling the output shaft torque of torque converter and the brake master cylinder pressure: such a low-level PID can be seen as a more realistic way to realize the powertrain time constant than the control-oriented vehicle model (\ref{standard dyanamics}).
The platoon is composed of 5 vehicles as in the previous experiments.
Refer to Table~\ref{table2} and Fig.~\ref{carsim e}(b) for the main parameters of the vehicles and refer to Table~\ref{table3} for the parameters of the PID controllers.

\begin{table}[h]
    \caption{Main CarSim parameters of the vehicles}
   \label{table2}
    \begin{center}
    \begin{tabular}{|c|c|c|c|}
    \hline
    $i$ & Type & Sprung mass [$kg$] & Tyres\\
    \hline
    $0$  & E-Class, Sedan & 1650 & 225/60 R18\\
    \hline
    $1$ & C-Class, Hatchback 2012 & 1270 & 225/15 R17\\
    \hline
    $2$ & Pickup, Compact: RWD & 1306 & 215/70 R15\\
    \hline
    $3$  & D-Class, Minivan 2012 & 1800 & 235/65 R17\\
    \hline
    $4$ & B-Class, Hatchback 2012 & 1110 & 185/65 R15\\
    \hline
    \end{tabular}
    \end{center}
\end{table}
\begin{table}[h]
    \caption{CarSim parameters of the PID controllers}
   \label{table3}
    \begin{center}
    \begin{tabular}{|c|c|c|c|}
    \hline
    Module & Proportional &  Integral & Derivative\\
    \hline
    Torque converter  & 0.5 & 0.2 & 0\\
    \hline
    \makecell{Brake master \\cylinder pressure} & 0.5 & 0.0005 & 0\\
    \hline
    \end{tabular}
    \end{center}
\end{table}
Fig.~\ref{carsim e} shows the spacing errors of each vehicle with non-adaptive controller and C-MRAC controller.
Even with the complex vehicle dynamics simulated by CarSim, Fig.~\ref{carsim e}(a) validates the theoretical results: the spacing errors in the non-adaptive case do not converge to zero, whereas the spacing errors with C-MRAC converge close to zero. 
The fact that convergence is not asymptotic is due to the fact that the CarSim vehicle dynamics are more complex than the linear dynamics (\ref{standard dyanamics}) used for control design: yet, one can say that C-MRAC achieves disturbance decoupling even in CarSim, thus validating the robustness of the results.

Fig.~\ref{carsim platoon} shows the longitudinal platoon in three conditions: (a) initial condition with $t=0\rm{s}$; (b) high velocity at time $t=19\rm{s}$; (c) low velocity at time $t=34\rm{s}$. 
Refer to Fig.~\ref{carsim platoon}(d) for accelerations and velocities of the platoon.
Fig.~\ref{carsim platoon}(a) shows that that the initial conditions are initially chosen such that the vehicles are very close: then, when the leading vehicle accelerates to a high velocity, Fig.~\ref{carsim platoon}(b) shows that the spacing error of each vehicle grows, and then the leader reaches a low velocity, Fig.~\ref{carsim platoon}(c) shows that the spacing error of each vehicle decreases. 
This is in accordance with the time headway spacing policy (\ref{spacing error}).

\section{Conclusions}
We proposed a novel design to longitudinal platooning in the framework of composite adaptation. 
The proposed design advances available platooning designs in three aspects: 
first, the condition for parameter convergence is relaxed from persistence of excitation of the acceleration or its derivative to finite excitation over a possibly short time interval; 
second, there is no need to store previous data samples nor to measure extra acceleration derivative states; 
third, it advances existing composite adaptation by dealing with uncertainty appears in the input matrix of the platooning error dynamics. 
Stability of the proposed design was studied analytically, and its robustness and practicality was verified with CarSim-based platooning experiments.
Interesting direction for future work is to study adaptation in the presence of discrete dynamics, such as gear switching  and/or communication losses\cite{c41,c42}. 
Further reducing the required measurements via output-feedback design is another possible future direction.

\bibliographystyle{IEEEtran}
\bibliography{reference}

\end{document}